\documentclass[11pt]{article}
\usepackage{fullpage,xspace}
\usepackage{amsmath,upgreek}
\usepackage{amsfonts}
\usepackage{amssymb}
\newtheorem{theorem}{Theorem}

\newtheorem{claim}[theorem]{Claim}

\newtheorem{corollary}[theorem]{Corollary}

\newtheorem{lemma}[theorem]{Lemma}

\newtheorem{proposition}[theorem]{Proposition}

\newcommand{\Xomit}[1]{ }
\newenvironment{proof}[1][Proof]{\textbf{#1.} }{\ \rule{0.5em}{0.5em}}

\newcommand{\eps}{\upvarepsilon}
\newcommand{\del}{\updelta}

\newcommand{\alp}{\upalpha}

\begin{document}

\date{}
\title{Online bin packing with cardinality constraints revisited}
\author{Gy\"{o}rgy D\'{o}sa\thanks{Department of Mathematics, University of Pannonia,
Veszprem, Hungary, \texttt{dosagy@almos.vein.hu}.}
\and Leah Epstein\thanks{ Department of Mathematics, University of Haifa,
Haifa, Israel. \texttt{lea@math.haifa.ac.il}. }}
\maketitle
\begin{abstract}
Bin packing with cardinality constraints is a bin packing problem where an upper bound $k \geq 2$ on the number of items packed into each bin is given, in addition to the standard constraint on the total size of items packed into a bin. We study the online scenario where items are presented one by one. We analyze it with respect to the absolute competitive ratio and prove tight bounds of $2$ for any $k \geq 4$. We show that First Fit also has an absolute competitive ratio of $2$ for $k=4$, but not for larger values of $k$, and we present a complete analysis of its asymptotic competitive ratio for all values of $k \geq 5$.
Additionally, we study the case of small $k$ with respect to the asymptotic competitive ratio and the absolute competitive ratio.
\end{abstract}

%

\section{Introduction}\label{intro}
Bin packing with cardinality constraints (BPCC) is a variant of
bin packing. The input consists of items, denoted by
$1,2,\ldots,n$, where item $i$ has a size $s_i>0$ associated with
it, and a global parameter $k \geq 2$, called the cardinality constraint.
The goal is to partition the input items into subsets called bins,
such that the total size of items of one bin does not exceed $1$,
and the number of items does not exceed $k$. In many applications
of bin packing the assumption that a bin can contain any number of
items is not realistic, and bounding the number of items as well
as their total size provides a more accurate modeling of the
problem. BPCC was studied both in the offline and online
environments \cite{KSS75,KSS77,KP99,CKP03,BCKK04,Epstein05,EL07afptas,FK13}.

In this paper we study online algorithms that receive and pack input items one by one, without any information on further input items. A fixed optimal offline algorithm that receives the complete list of items before packing it is denoted by $OPT$. For an input $L$ and algorithm $A$, we let $A(L)$ denote the number of bins that $A$ uses to pack $L$. We also use $OPT(L)$ to denote also the number of bins that $OPT$ uses for a given input $L$.
The absolute approximation ratio of an algorithm $A$ is the supremum ratio over all inputs $L$ between the number of bins $A(L)$ that it uses and the number of bins $OPT(L)$ that $OPT$ uses. The asymptotic approximation ratio is the limit of absolute competitive ratios $R_K$ when $K$ tends to infinity and $R_K$ takes into account only inputs for which $OPT$ uses at least $K$ bins. Note that (by definition), for a given algorithm (for some online bin packing problem), its  asymptotic approximation ratio never exceeds its absolute approximation ratio. If the algorithm is online, the term {\it competitive ratio} replaces the term {\it approximation ratio}. For an algorithm whose approximation ratio (or competitive ratio) does not exceed $R$, we say that it is an $R$-approximation (or $R$-competitive). We see a bin as a set of items, and for a bin $B$, we let $s(B)=\sum_{i \in B} s_i$ be its level.

Bin packing problems are often studied with respect to the
asymptotic measures. Approximation algorithms were designed for
the offline version (that is strongly NP-hard for $k \geq 3$)
\cite{KSS75,KP99,CKP03,EL07afptas}, and the problem has an
asymptotic fully polynomial approximation scheme (AFPTAS)
\cite{CKP03,EL07afptas}. Using elementary bounds, it was shown by
Krause, Shen, and Schwetman \cite{KSS75} that the cardinality
constrained variant of First Fit (FF), that packs an item $i$ into a
minimum index bin where it fits both with respect to size and
cardinality (i.e., it has at most $k-1$ items and level at most $1-s_i$), has an asymptotic competitive ratio of at most $2.7-\frac{2.4}k$. For $k \rightarrow \infty$, the competitive
ratio is $2.7$, since this is a special case of vector bin packing (with two dimensions) \cite{GareyGJ76}. The case $k=2$ is solvable using matching techniques in the
offline scenario, but it is not completely resolved in the online
scenario. Liang \cite{Liang80} showed a lower bound of $\frac 43$
on the asymptotic competitive ratio  for this case, Babel et al.~\cite{BCKK04} improved the lower bound to $\sqrt{2}\approx
1.41421$, and designed an algorithm whose asymptotic competitive
ratio is at most $1+\frac 1{\sqrt{5}}\approx 1.44721$ (improving
over FF). Recently, Fujiwara and Kobayashi \cite{FK13} improved
the lower bound to $1.42764$. For larger $k$, there is a
$2$-competitive algorithm \cite{BCKK04}, and improved algorithms
are known for $k=3,4,5,6$ (whose competitive ratios are at most
$1.75$, $1.86842$, $1.93719$, and $1.99306$, respectively)~\cite{Epstein05}. Note that the upper bounds of \cite{KSS75} for
FF and $k=3$ is $1.9$, and an algorithm whose competitive ratio is at
most $1.8$ was proposed by \cite{BCKK04}. A full analysis of the
cardinality constrained variant of the Harmonic algorithm
\cite{LeeLee85} (that partitions items into $k$ classes and packs
each class independently, such that the classes are
$I_{\ell}=(\frac 1{\ell+1},\frac 1\ell]$ for $1 \leq \ell \leq k -1$
and $I_k=(0,\frac 1k]$, and for any $1 \leq \ell \leq k$, each bin of $I_{\ell}$, possibly except for
the last such bin, receives exactly $\ell$ items) is given in
\cite{Epstein05}, and its competitive ratio for $k=2,3$ is $1.5$
and $\frac{11}6$, respectively (its competitive ratio is in $[2,2.69103]$ for $k \geq
4$, see Table \ref{tabtab} for some additional values). As for lower bounds, until recently, except for the case
$k=3$ for which a lower bound of $1.5$ on the competitive ratio
was proved in \cite{BCKK04}, most of the known lower bounds
followed from the analysis of lower bounds for standard bin
packing \cite{Yao80A,Vliet92,BBG}. New lower bounds for many
values of $k$ were given in \cite{FK13}, and in particular, they
proved lower bounds of $1.5$ and $\frac{25}{17}\approx 1.47058$
for $k=4$ and $k=5$, respectively. For $6 \leq k \leq 9$, the
current best lower bound remained $1.5$, that was implied by the
lower bound of Yao \cite{Yao80A}, and for $k=10$ and $k=11$, lower
bounds of $\frac{80}{53}\approx 1.50943$ and $\frac{44}{29}\approx
1.51724$, respectively, were proved in \cite{FK13} (see \cite{FK13} for the
lower bounds of other values of $k$). In this paper we provide a complete analysis of FF with respect to the asymptotic competitive ratio. We find that its competitive ratio is at most $2.5-\frac 2k$ for $k=2,3,4$,  $\frac{8(k-1)}{3k}=\frac 83-\frac{8}{3k}$ for $4 \leq k \leq 10$, and  $2.7-\frac 3k$ for $k \geq 10$ (we included each of the values $k=4$ and $k=10$ in two cases as $2.5-2/k=8/3-8/(3k)$ for $k=4$, and $8/3-8/(3k)=2.7-3/k$ for $k=10$).
Additionally, we provide improved lower bounds on the asymptotic competitive ratio of arbitrary online algorithms for $k=5,7,8,9,10,11$. The values of these lower bounds are $1.5$ for $k=5$, and approximately  $1.51748$, $1.5238$, $1.5242$, $1.526$, $1.5255$, for $k=7,8,9,10,11$, respectively.

There are few known results for the absolute measures. The asymptotic $(1+\eps)-$approximation algorithm of Caprara, Kellerer, and Pferschy \cite{CKP03} uses $(1+\eps)OPT+1$ bins to pack the items, and thus, choosing $\eps>0$ to be small (for example $\eps=\frac{1}{100}$) results in a polynomial time  absolute $\frac 32$-approximation algorithm. This is the best possible unless P=NP. In the online environment, it is not difficult to see that given the absolute upper bound of $1.7$ on the competitive ratio of FF for standard bin packing \cite{DS12}, the upper bound of $2.7-2.4/k$ becomes an absolute one (we provide the proof here for completeness). In this paper we analyze the absolute competitive ratio, and show a tight bound of $2$ on the absolute competitive ratio for any $k \geq 4$. The upper bound for $k=4$ is proved for FF. An upper bound for $k=5$ is proved using an algorithm that performs FF except for one case. We show that a variant of the algorithm of \cite{BCKK04} has an absolute competitive ratio $2$ for any $k \geq 2$. In the case $k=3$, we provide a lower bound of $\frac 74=1.75$ on the absolute competitive ratio of any algorithm, and show that the absolute competitive ratio of FF is $\frac {11}6 \approx 1.8333$. For $k=2$, tight bounds of $1.5$ on the best possible competitive ratio follow from previous work \cite{Bek} (the upper bound is the absolute competitive ratio of FF).

For standard bin packing
\cite{Ullman71,John,Johnso74,JDUGG74,LeeLee85,Yao80A,RaBrLL89}, it
is known that the asymptotic competitive ratio is in
[1.5403,1.58889] \cite{BBG,Seiden02J}, and the absolute
competitive ratio is in $[\frac 53,1.7]$ (see \cite{Zhang,DS12}),
where the upper bound is the competitive ratio of FF (without
cardinality constraints). Interestingly, introducing cardinality
constraints (with sufficiently large values of $k$) results in an
increase of many competitive ratios by $1$
\cite{KSS75,JDUGG74,LeeLee85,Epstein05}. Another related problem
is called {\it class constrained bin packing}
\cite{EIL10,ST01,ST04,XM06}. In that problem each item has a
color, and a bin cannot contain items of more than $k$ colors (for
a fixed parameter $k$). BPCC is the special case of that problem
where all items have distinct colors.

We start with lower bounds in Section \ref{lbproofs}, where both the absolute competitive ratio and the asymptotic competitive ratio are studied. We consider algorithms afterwards, in Section \ref{algs}, where we analyze FF, and  in Section \ref{oth}, where we consider algorithms whose absolute competitive ratio is at most $2$.

{\begin{table}[h!]
\renewcommand{\arraystretch}{1.05   }

$$
\begin{array}{||c|c|c|c|c|c|c||}
\hline
\hline

\mbox{Value of \  } k & \mbox{prev. LB} &  \mbox{new LB} & \mbox{FF} & \mbox{prev. UB for FF} & \mbox{Harmonic} & \mbox{best known} \\

\hline

2 & 1.42764 \  \cite{FK13} & -  & 1.5 & 1.5 & 1.5 & 1.44721 \ \cite{BCKK04}\\

\hline

3 & 1.5 \  \cite{BCKK04} &  - &  1.8333 & 1.9 & 1.8333 & 1.75 \ \cite{Epstein05} \\

\hline

4& 1.5 \  \cite{FK13} &  - &  2 & 2.1  &  2 &  1.86842 \ \cite{Epstein05}\\
\hline

5& 1.47058 \  \cite{FK13} &  1.5  & 2.1333 & 2.22 & 2.1  & 1.93719 \ \cite{Epstein05}\\
\hline
6 &  1.5 \  \cite{Yao80A} & -  & 2.2222 & 2.3 &  2.16667 &  1.99306 \ \cite{Epstein05}\\
\hline
7 & 1.5  \  \cite{Yao80A} &  1.5174825 &  2.2857 & 2.35714& 2.2381& 2\ \cite{BCKK04}   \\
\hline
8 & 1.5  \  \cite{Yao80A} &  1.5238095 & 2.3333 & 2.4 & 2.29167   &2\ \cite{BCKK04}  \\
\hline
9 & 1.5  \  \cite{Yao80A} & 1.52419355  & 2.3704 & 2.43333& 2.33333&    2 \ \cite{BCKK04}  \\
\hline
10 & 1.50943 \  \cite{FK13} &  1.525974 &  2.4 & 2.46& 2.36666 & 2  \ \cite{BCKK04}   \\
\hline
11 & 1.51724 \  \cite{FK13} &  1.525547 &  2.4273  & 2.481818 & 2.39394   & 2 \ \cite{BCKK04}  \\
\hline
12 & 1.53488  \  \cite{FK13} & -  & 2.45  & 2.5&  2.41667  & 2 \ \cite{BCKK04}  \\
\hline
\hline
\end{array}
$$
\caption{\label{tabtab} Bounds for $2 \leq k \leq 12$. The first column contains the previously known lower bounds on the asymptotic competitive ratio. The second column contains our improved lower bounds. The third column contains the tight asymptotic competitive ratio of FF (for $k=2,3,4$ it is also the absolute competitive ratio), the fourth column contains the previous upper bound on FF's asymptotic competitive ratio \cite{KSS75}, the fifth column contains the tight asymptotic competitive ratio of Harmonic \cite{Epstein05}, and the last column contains the asymptotic competitive ratio of the current best algorithm.}
\end{table}
}

\section{Lower bounds}\label{lbproofs}
In this section we present lower bounds for the two measures.

\subsection{Lower bounds on the absolute competitive ratio}
We show that the absolute competitive ratio is at least $2$ for $k \geq 4$. Together with the analysis of Section \ref{Kot}, we will find that this is the best possible competitive ratio.

\begin{proposition}
The absolute competitive ratio of any algorithm for $k \geq 4$ is at least $2$.
\end{proposition}
\begin{proof}
Let $0<\eps<\frac 1{3k}$.
The input starts with $k$ items, each of size $\eps$, called tiny items. Since an optimal solution packs them into one bin, if an online algorithm uses two bins, then we are done. Otherwise the algorithm packs them into one bin, and no further items can be combined into this bin, since it already has $k$ items. The next two items have sizes of $\frac 13+\eps$. If the algorithm packs them into two new bins, then the next item has size $\frac 23$ and it requires a new bin. An optimal solution packs the last item with $k-1$ tiny items, and the remaining three items into another bin, while the algorithm uses four bins.
Otherwise, the algorithm packs the two items of sizes $\frac 13+\eps$ into one bin. In this case the last two items have sizes of  $\frac 12+\eps$. The algorithm now has four bins, while an optimal solution has two bins, each with an item of size $\frac 12+\eps$, an item of size $\frac 13+\eps$, and $\lfloor \frac k2 \rfloor$ or $\lceil \frac k2 \rceil$ items of size $\eps$ each (which is possible since at most $k-2$ items of size $\eps$ are added; for $k=4$  $\frac k2=k-2=2$ holds, and $\lceil \frac k2 \rceil \leq \frac{k+1}2 \leq k-2$ holds for $k \geq 5$).
\end{proof}

In the case $k=2$, a lower bound of $\frac 32$  on the absolute
competitive ratio follows from an input that consists of two
items, each of size $\eps$, possibly followed by two items, each
of size $1-\eps$ (for $0<\eps<\frac 12$) \cite{Bek}.

Next, we present a lower bound of $\frac 74$ on the absolute
competitive ratio of any algorithm for $k=3$. Recall that the best
asymptotic competitive ratio for $k=3$ is in $[\frac 32,\frac
74]$. The upper bound of \cite{KSS75} for the asymptotic
competitive ratio of FF is $2.7-2.4/3=1.9$, and we will show a
tight bound of $\frac{11}{6}$ on the absolute and asymptotic
competitive ratios of FF.
\begin{proposition}
The absolute competitive ratio of any algorithm for $k=3$ is at least $\frac 74=1.75$.
\end{proposition}
\begin{proof}
Let $0<\eps<\frac 1{24}$. 
The input starts with three tiny jobs of size $\eps$ each. Since an optimal solution can pack them into one bin, to avoid a competitive ratio of at least $2$, the algorithm must do the same. Note that the bin containing these items cannot receive any additional items. Next, two items of sizes $\frac 13+\eps$ arrive. If the last two items are packed into separate bins, an item of size $\frac 23$ is presented. An optimal solution can pack the items into two bins; the last item is combined with two tiny items, and the remaining three items are packed into a second bin, which is possible given the value of $\eps$.

Otherwise, the algorithm has two bins, where the second one can still receive one item of size at most $\frac 13-2\eps$. The remaining items will be larger, and thus they will be packed into new bins. Now, two items of sizes $\frac 13+3\eps$ arrive. If the algorithm uses two new bins to pack them, then two items whose sizes are equal to $\frac 23-2\eps$ arrive, and the algorithm is forced to use two new bins for them, for a total of six bins. An optimal solution uses three bins; two bins contain (each) an item of size $\frac 23-2\eps$, an item of size $\frac 13+\eps$, and a tiny item. The remaining three items have total size $\frac 23+7\eps$ and can be packed into a third bin by an optimal solution. Thus, the competitive ratio is $2$ in this case.

Otherwise, the algorithm has three bins, where the second and third bins can still receive one item each, but no item of size at least $\frac 13$ can be packed there. The remaining items will be larger than $\frac 13$, and thus they will be packed into new bins. Specifically, there are four items whose sizes are equal to $\frac 23-4\eps>\frac 12$. Each such item must be packed into a new dedicated bin, for a total of seven bins. An optimal solution can combine each such item with an item of size $\frac 13+\eps$ or an item of size $\frac 13+3\eps$, and three such bins also receive a tiny item.
\end{proof}


\subsection{Lower bounds on the asymptotic competitive ratio}
In this section, we present improved lower bounds for  ratio for ${k=5}$ and ${7\leq k\leq 11}$.
To prove these lower bounds, we will consider inputs that consist of at most four batches.
Let $0<\del<\frac{1}{2000}$,
and let $N$ be a large integer divisible by $6k$. For $7 \leq k \leq 11$, the first batch contains
$\frac{k-6}6 \cdot N$ items, and for $k=5$, the first batch contains $\frac N2$ items, each of size $\frac{1}{42}-3\del>0$. We use $\phi_k=\frac{k-6}{6k}$ for $7 \leq k \leq 11$, and $\phi_5=\frac 1{10}$. Thus, the first batch consists of $ k \phi_k \cdot N$ items.
The second batch contains $N$ items, each of size
$\frac{1}{7}+\del$, the third batch contains $N$ items, each of
size $\frac{1}{3}+\del$, and the fourth batch contains $N$ items,
each of size $\frac{1}{2}+\del$. Note that the first batch has a
smaller number of items than later batches for all values of $k$. The input may stop
after any batch, and thus there are four possible inputs, denoted
by $L_1$, $L_2$, $L_3$, and $L_4$. For $i=1,2,3,4$, an optimal
solution for $L_i$ has bins packed almost identically.
For $k=5$, clearly,
$OPT(L_1)\geq \frac{N}{10}$  and $OPT(L_2)\geq \frac{3N}{10}$
as no bin can have more than
$k$ items, $OPT(L_3) \geq \frac N2$, as no bin
can have more than two items of the third batch, and $OPT(L_4)
\geq N$, as no bin can have more than one items of the fourth
batch.
For $7 \leq k \leq 11$, clearly,
$OPT(L_1)\geq \frac{N\cdot(k-6)}{6k}$ as no bin can have more than
$k$ items, $OPT(L_2) \geq \frac N6$, as no bin can have more than
six items of the second batch, $OPT(L_3) \geq \frac N2$, as no bin
can have more than two items of the third batch, and $OPT(L_4)
\geq N$, as no bin can have more than one items of the fourth
batch.
Next, we define solutions that achieve those numbers and
thus they are optimal.  For $k=5$, optimal solutions for $L_1$ and for $L_2$ have five
items in each bin, which is possible since $5 \cdot (\frac 1{7}+\del) =\frac 57+5\del<1$. Thus, $OPT(L_1)=\frac {N}{10}$ and $OPT(L_2)=\frac {3N}{10}$. An optimal
solution for $L_3$ has two items of the third batch, two items of
the second batch, and one item of the first batch in each
bin, $2(\frac
13+\del)+2(\frac{1}{7}+\del)+(\frac{1}{42}-3\del)<1$, all
items are packed, and $OPT(L_3)=\frac {N}{2}$. Finally, an optimal
solution for $L_4$ has at most one item of the first batch and
exactly one item of any other batch packed into each bin. All
items are packed, $OPT(L_4)=N$, and $(\frac 12+\del)+(\frac
13+\del)+(\frac 17+\del)+(\frac 1{42}-3\del)=1$.
For $7 \leq k \leq 11$, an optimal solution for $L_1$ has $k$
items in each bin, which is possible since $k \cdot (\frac 1{42}-3\del) < \frac
{11}{42}<1$. Thus, $OPT(L_1)=\frac {N(k-6)/6}{k}$. An optimal solution
for $L_2$ has six items of the second batch and $k-6$ items of the
first batch in each bin, and $OPT(L_2)=\frac {N}{6}$. The solution
is valid as each bin has $k$ items,
$6(\frac{1}{7}+\del)+(k-6)(\frac{1}{42}-3\del) = \frac
67+\frac{k-6}{42}+(24-3k)\del \leq
\frac{6}{7}+\frac{5}{42}+24\del=\frac{41}{42}+24\del<1$, as $\del < \frac{1}{2000}$, and
all the items of the first two batches are packed. An optimal
solution for $L_3$ has two items of the third batch, two items of
the second batch, and at most two items of the first batch in each
bin, $2(\frac
13+\del)+2(\frac{1}{7}+\del)+2(\frac{1}{42}-3\del)=1-2\del<1$, all
items are packed, and $OPT(L_3)=\frac {N}{2}$. Finally, an optimal
solution for $L_4$ has at most one item of the first batch and
exactly one item of any other batch packed into each bin. All
items are packed, $OPT(L_4)=N$, and $(\frac 12+\del)+(\frac
13+\del)+(\frac 17+\del)+(\frac 1{42}-3\del)=1$. In all cases we have $OPT(L_1)=N\cdot \phi_k$. We use $\phi'_k=\frac{OPT(L_2)}N$, i.e., $\phi'_5=\frac{3}{10}$, and $\phi'_k=\frac 16$ for $7 \leq k \leq 11$.

Consider a deterministic or randomized algorithm $A$.  Let $X_i$
be the number of bins (or expected number of bins) that the
algorithm opens while packing the items of batch $i$. Assume that the competitive ratio
is $R$. Let $f$ be a function where $f(n)=o(n)$ such that for any
input $I$, $A(I) \leq R \cdot OPT(I)+f(OPT(I))$. We have $A(L_i)
\leq R \cdot OPT(L_i)+f(OPT(L_i))$. Note that $A(L_i)=\sum_{j=1}^i
X_j$. For any set of four parameters $\alp_i\geq 0$ for
$i=1,2,3,4$ (constants that are independent of $N$), we find
$\sum_{i=1}^4 \alp_i A(L_i) \leq \sum_{i=1}^4 \alp_i(R \cdot
OPT(L_i)+f(OPT(L_i)))$. Letting $\beta_i=\sum_{j=i}^4 \alpha_i$
and rewriting it we get $\sum_{i=1}^4 \beta_i X_i \leq RN \cdot
(\alpha_1 \phi_k+\frac{\alpha_2}\phi'_k+\frac{\alpha_3}2+\alpha_4)+
 \sum_{i=1}^4 \alpha_i f(OPT(L_i))$.

We define weights for the items. Let $w_i \geq 0$ be the weight of
an item of batch $i$.  The total weight of items of $L_4$ is $w_1 \cdot k\phi_k
\cdot N + w_2 \cdot N+ w_3 \cdot N+ w_4 \cdot N$. Let $W_i$ denote
the maximum weight of any bin opened by the algorithm for batch
$i$ or a later batch (possibly used for additional items later). We have $W_1 \geq W_2 \geq W_3 \geq W_4$.  The total weight of items is, denoted by $W$ satisfies $W \leq \sum_{i=1}^4 W_iX_i$. Thus $W=w_1 \cdot k\phi_k \cdot N + w_2 \cdot N+ w_3 \cdot N+ w_4 \cdot N = W \leq \sum_{i=1}^4 W_iX_i $.
Let $W_5=0$, and $\beta_i=W_i$ (i.e., $\alpha_i=W_i-W_{i+1}$ for $i=1,2,3,4$).

\begin{lemma}
For $ 7 \leq k \leq 11$, we have $R \geq \frac{\frac{k-6}6 w_1
 + w_2 + w_3 + w_4}{W_1/6-W_1/k+W_2/k+W_3/3+W_4/2}$. For $k=5$ we have $R \geq \frac{w_1/2
 + w_2 + w_3 + w_4}{W_1/10+W_2/5+W_3/5+W_4/2}$.
\end{lemma}
\begin{proof}
We have $N(w_1 \cdot k\phi_k + w_2 + w_3 + w_4)  \leq \sum_{i=1}^4 \beta_iX_i \leq RN((W_1-W_2)\phi_k+(W_2-W_3)\phi'_k+\frac{W_3-W_4}2+W_4)+ \sum_{i=1}^4 (W_i-W_{i+1}) f(OPT(L_i))$.

We find $\sum_{i=1}^4 (W_i-W_{i+1}) f(OPT(L_i))=o(N)$, since $OPT(L_i) =\Theta(N)$ for $i=1,2,3,4$, and the values $W_i$ are constants independent of $N$.

Thus, $w_1 \cdot k\phi_k  + w_2 + w_3+ w_4 \leq  R((W_1-W_2)\phi_k+(W_2-W_3)\phi'_k+\frac{W_3-W_4}2+W_4)$, or alternatively, $w_1 \cdot k\phi_k  + w_2 + w_3+ w_4 \leq  R(W_1\phi_k+W_2(\phi'_k-\phi_k)+W_3(\frac 12-\phi'_k)+\frac{W_4}2)$.
\end{proof}

For $k=5$, let $w_2=1$ and $w_1=w_3=w_4=2$.
For $k=7,8$, let $w_1=1$, $w_2=1$, $w_3=w_4=2$. For $k=9$, let
$w_1=1$, $w_2=2$, $w_3=w_4=4$. For $k=10,11$, let $w_1=1$,
$w_2=3$, $w_3=w_4=6$.

\begin{lemma}
For $k=5$, we have $W_1\leq 10$, $W_2 \leq 6$, $W_3 \leq 4$, and $W_4
\leq 2$.

For $k=7,8$, $W_1\leq k+2$, $W_2 \leq 6$, $W_3 \leq 4$, and $W_4
\leq 2$.

For $k=9$, $W_1\leq 16$, $W_2 \leq 12$, $W_3 \leq 8$, and $W_4
\leq 4$.

For $k=10,11$, $W_1\leq k+12$, $W_2 \leq 18$, $W_3 \leq 12$, and
$W_4 \leq 6$.
\end{lemma}
\begin{proof}
The claim regarding $W_4$ holds since $W_4=w_4$ must hold (as
every bin opened for the last batch contains exactly one
item).
Moreover, any bin opened for the third batch will contain only
items of sizes strictly above $\frac 13$, so it can contain at
most two items and since $w_3=w_4$ in all cases, $W_3 \leq 2w_3$.

Let $k=5$. Consider a bin that was opened for the second batch. If it has two items of later batches, it can have at most two items of the second batch, so the total weight is at most $2w_2+2w_3=6$. If it has at most one item of later batches, then all its items have weight $1$ except for at most one item whose weight is at most $2$, for a total of $6$, since there are at most five items in total.
Finally, consider a bin that was opened for the first batch. Every item has weight of at most $2$, and there are at most five items, thus the total weight is at most $10$.

We are left with $7 \leq k \leq 11$, and bounding $W_1$ and $W_2$.
A bin that was opened for the second batch can contain at most six
items of that batch. If it contains one larger item, then it can
contain at most four items of the second batch, and if it contains
two larger items, then it can contain at most two items of the
second batch. It cannot contain more than two larger items, and
since $w_3=w_4$, we do not distinguish the items of the last two
batches. Thus $W_2 \leq \max\{6w_2,w_3+4w_2,2w_3+2w_2\}=6w_2$. We find
$W_2=6$ for $k=7,8$, and $W_2=12$ for $k=9$, and $W_2=18$ for
$k=10,11$.

Finally, consider bins opened for the first batch. We start with the
case $k=7,8$. A bin contains at most $k$ items, where items of the
last two batches have weights of $2$ and items of the first two
batches have weights of $1$. Thus, the total weight is at most $k$
plus the number of items of sizes above $\frac 13$, which is at
most $2$. The total weight is therefore at most $k+2$.

Consider the cases $k=9,10,11$. Given a bin that contains at most
three items of batches $2,3,4$, it can contain at most two items
of batches $3,4$, so the total weight is at most
$k+2(w_3-1)+(w_2-1)$. This last value is equal to $2w_3+w_2+k-3=5w_2+k-3$,
which is equal to $16$ for $k=9$ and to $k+12$ for $k=10,11$. We
are left with the case that there are at least four items of
batches $2,3,4$ packed into the bin. If all those items are of batch $2$, then there
are at most six such items, and the total weight is at most
$k+6(w_2-1)=6w_2+k-6 \leq 5w_2+k-3$ as $w_2 \leq 3$. If there is one item of the last two
batches, there are at most four items of batch $2$. If there are
at most three items of batch $2$, then the weight is at most
$k+3(w_2-1)+(w_3-1)=3w_2+w_3+k-4=5w_2+k-5$. If there is
one item of the last two batches and four items of the second
batch, their total size is at least $\frac 13+\del+4(\frac
17+\del)=\frac{19}{21}+5\del$. The remaining space can only
contain four items of the first batch as $5(\frac
1{42}-3\del)+\frac{19}{21}+5\del=1+\frac{1}{42}-10\del>1$. The
total weight is therefore at most $4w_1+4w_2+w_3$, which is equal
to $16$ for $k=9$, and to $22 \leq k+12 $ for $k=10,11$.

If there are two items of the last two batches, there are at most
two items of batch $2$. If there is at most one item of batch $2$,
then the weight is at most
$k+(w_2-1)+2(w_3-1)=w_2+2w_3+k-3=5w_2+k-3$. If there are two items of the last two
batches and two items of the second batch, their total size is at least
$2(\frac 13+\del)+2(\frac 17+\del)=\frac{20}{21}+4\del$. The
remaining space can only contain two items of the first batch as
$3(\frac 1{42}-3\del)+\frac{20}{21}+4\del=1+\frac{1}{42}-5\del>1$.
The total weight is therefore at most $2w_1+2w_2+2w_3$, which is
equal to $14$ for $k=9$, and to $20 \leq k+12 $ for $k=10,11$.
\end{proof}

\begin{corollary}
The following values are lower bounds on the competitive ratios.
\begin{itemize}
\item $\frac 32=1.5$ for $k=5$.
\item $\frac{k^2+24k}{k^2+10k+24}$ for $k=7,8$. This value is equal to $217/143 \approx 1.5174825$ for $k=7$ and to $\frac{32}{21} \approx 1.5238095$ for $k=8$.
\item $\frac{10.5}{62/9}=\frac{189}{124} \approx 1.52419355$, for $k=9$.
\item $\frac{k^2+84k}{k^2+48k+36}$ for $k=10,11$. This value is equal to $235/154 \approx 1.525974$ for $k=10$ and to $\frac{209}{137} \approx 1.525547$ for $k=11$.
\end{itemize}
\end{corollary}

Note that in the cases $k=6$ and $k=12$, our methods do not produce improved lower bounds, and they give exactly the known lower bound.

\section{A complete analysis of First Fit}\label{algs}
We provide a complete analysis of the asymptotic competitive
ratio. For $k=2,3,4$, the bounds that we find are the absolute
competitive ratios as well. In the analysis, a bin of FF that has
$j$ items for $j \leq k$ is called  a $j$-bin, and a bin whose
number of items is in $[j,k-1]$ for some $1\leq j<k$ is called a
$j^+$-bin.

We find that the asymptotic competitive ratio of FF is $2.5-\frac
2k$ for $k=2,3,4$,  $\frac{8(k-1)}{3k}=\frac 83-\frac{8}{3k}$ for
$4 \leq k \leq 10$, and  $2.7-\frac 3k$ for $k \geq 10$ (recall that the values $k=4$ and $k=10$ are included in two cases each). The
values for $k=2,3,\ldots,12$ are given in Table{tabtab}. An interesting property is that for large values
of $k$ ($k$ tending to infinity) both the competitive ratio of the
cardinality constrained Harmonic algorithm and FF have competitive
ratios that are larger by $1$ than their competitive ratios for
standard bin packing. Thus, Harmonic has a slightly smaller
competitive ratio of $1.69103$. Moreover, it can be verified that
the worst-case examples of Harmonic are valid (but not tight) for FF. For
$k=2,3,4$ they have the same competitive ratios, but not for
$k\geq 5$, and in many cases the competitive ratio of Harmonic is
much smaller (see examples in Table \ref{tabtab}).

We start with examples showing that the asymptotic competitive
ratios cannot be smaller. For $2\leq k \leq 4$, let ${\ell} \geq 0$ be a large integer,
and let $0 < \eps < \frac 1{9k}$. Consider an input consisting of
$2k(k-2){\ell}$ items of size $\eps$ each (smallest items),
$2k{\ell}$ items of size $\frac 12-k\eps>\frac 13$ each (medium
size items), and $2k{\ell}$ items of size $\frac 12+\eps$ each
(largest items). The items are presented in this order. FF creates
$2(k-2){\ell}$ bins containing $k$ smallest items each. Then, as
further items are larger than $\frac 13$, FF creates $k{\ell}$
bins containing pairs of medium size items, and as the remaining
items are larger than $\frac 12$, the largest items are packed
into $2k{\ell}$ dedicated bins. For this input $L_{\ell}$,
$OPT(L_{\ell})=2k{\ell}$, since it is possible to pack a largest
item, a medium size item, and $k-2$ smallest items into a bin as
$\frac 12+\eps +\frac 12 -k\eps+ (k-2)\eps < 1$, while
$FF(L_{\ell})=2(k-2){\ell}+k{\ell}+2k{\ell}=5k{\ell}-4{\ell}$.
This shows that the asymptotic competitive ratio of FF is at least
$2.5-\frac 2k$, that is, at least $\frac{11}6$ for $k=3$ and at
least $2$ for $k=4$. The example is valid for $k=2$ too, giving
the value $1.5$ (in this case there are no smallest items).

In the case $5 \leq k \leq 10$, let $\ell$ be a positive integer
divisible by $k$, let $0< \eps < \frac 1{120}$ and $\delta <
\frac{\eps}{3^{\ell+4}}$ be small positive values, and consider
the following input. There are $3\ell$ items of size $\frac
12+\del$, $\ell$ items of size $\frac 12-10\del$, $\ell$ items of
size $\frac 14+20\del$, $\ell$ items of size $\frac 14-30\del$,
$(3k-8)\ell$ items of size $\del$, and for $1\leq p\leq \ell$
there is a pair of items of sizes $\frac 14+\frac{\eps}{3^p}$ and
$\frac 14-\frac{\eps}{3^p}-10\del$. Since $\del<\eps<\frac
1{120}$, all sizes are strictly positive. An optimal solution has
three types of bins. There are $\ell$ bins with an item of size
$\frac 12+\del$, an item of size $\frac 12-10\del$, and $k-2 \leq
8$ items of size $\del$ each. There are $\ell$ bins with an item
of size $\frac 12+\del$, an item of size $\frac 14+20\del$,  an
item of size $\frac 14-30\del$ and $k-3 \leq 7$ items of size
$\del$ each. Finally, there are $\ell$ bins, where the $p$th bin
has an item of size $\frac 12+\del$, the pair of items of sizes
$\frac 14+\frac{\eps}{3^p}$ and $\frac
14-\frac{\eps}{3^p}-10\del$, and $k-3 \leq 7$ items of size $\del$
each. Remove the items of sizes $\frac 14+\frac{\eps}{3^{\ell}}$
and $\frac 14-10\del-\frac{\eps}{3}$, and one item of size $\frac
14-30\del$ from the input. Obviously, an optimal solution still
requires at most $3\ell$ bins. For $1 \leq p \leq \ell-1$, the
items of sizes  $\frac 14+\frac{\eps}{3^{p}}$ and $\frac
14-10\del-\frac{\eps}{3^{p+1}}$ are called a modified pair of
index $p$.

The items are presented to FF in the following order. First, all
items of size $\del$ are presented and packed into
$(3k-8)\frac{\ell}k$ bins that cannot be used again. Next, for $1
\leq p \leq \ell-1$, the modified pair of items of index $p$ is
presented, followed by an item of size $\frac 14-30\del$. The
total size of these three items is $\frac
14+\frac{\eps}{3^{p}}+\frac 14-10\del-\frac{\eps}{3^{p+1}}+\frac
14-30\del=\frac 34-40\del+\frac{2\eps}{3^{p+1}}>\frac
34-\frac{40\eps}{3^{\ell+4}}+\frac{2\eps}{3^{p+1}}\geq \frac
34-\frac{40\eps}{3^{p+5}}+\frac{2\eps}{3^{p+1}}>\frac
34+\frac{\eps}{2\cdot 3^{p}}$, while further items have sizes of
$\frac 12+\del>\frac 12-10\del>\frac
14+20\del>\frac 14$, $\frac 14+\frac{\eps}{3^{p'}}>\frac 14$, $\frac 14-30\del > \frac 14-10\del-\frac{20\eps}{3^{\ell+4}} \geq \frac 14-10\del-\frac{\eps}{3^{\ell+1}} > \frac 14-10\del-\frac{\eps}{3^{p'+1}}$, and
$\frac 14-10\del-\frac{\eps}{3^{p'+1}}$, where $p'\geq p+1$ (further modified pairs exist only if $p<\ell-1$). We have $\frac
34+\frac{\eps}{2\cdot 3^{p}}+\frac
14-\frac{\eps}{3^{p'+1}}-10\del\geq 1+\frac{\eps}{2\cdot
3^{p}}-\frac{\eps}{3^{p+2}}-\frac{10\eps}{3^{p+5}}=1+\frac{\eps(
3^5/2-3^3-10)}{3^{p+5}}>1$. This proves that after a bin of a
modified pair and an item of size $\frac 14-30\del$ is created, no
further items can be packed into that bin. When no modified pairs
remain, pairs of items of sizes $\frac 12-10\del$ and $\frac 14
+20\del$ are presented (there are $\ell$ such pairs). Each bin
receives such a pair, whose total size is $\frac 34+10\del$. Since
all remaining items have sizes above $\frac 14$, each created bin
will not be used for further items. Finally, all remaining items
(of sizes $\frac 12+\del$) are packed into dedicated bins. The
total number of bins is
$(3k-8)\frac{\ell}k+\ell-1+\ell+3\ell=\frac{8k-8}{k}\ell-1$. Since
an optimal solution has at most $3\ell$ bins, we find that the
asymptotic competitive ratio is at least $\frac{8(k-1)}{3k}$.

In the case $k \geq 10$, we adapt the lower bound example of FF
\cite{JDUGG74} by adding tiny items. The original construction was such that almost every bin of OPT (all bins except for a constant number of bins) had an item whose size was $\frac 12+\del$.
We replace those items with items of sizes $\frac 12+\frac{\del}2$, and add $k-3$ tiny items of sizes $\frac {\del}k$ to each such bin. All bins of OPT that had an item of size $\frac 12+\del$ have at most three items each, so tiny items can be added to almost all bins of OPT, and this modification keeps those bins valid. The tiny items are presented to FF before other items, so they are packed into bins containing $k$ items each, that cannot be used for other items. The items of sizes $\frac 12+\frac{\del}2$ are presented last and must be packed into dedicated bins, as any previous bin either has $k$ items or total size above $\frac 12$. Thus, the modified construction will give a lower bound of $1.7+\frac{k-3}k$ on the asymptotic competitive ratio of FF. We describe the exact construction for completeness.

Let $\ell$ be a
positive integer, such that $\ell-1$ is divisible by $k$ and by
$k-3$, let $0< \eps < \frac 1{120}$ and $\delta <
\frac{\eps}{3^{\ell+4}}$ be small positive values, and consider
the following input. The instance consists of $10(k-3)(\ell-1)$
tiny items of size $\frac{\delta}k$ and $30\ell$ larger items. We
describe the items and the packing of FF, where the larger items
are packed into $17\ell$ bins: $B_1,\ldots,B_{17\ell}$. The tiny
items are presented first, and they are packed into
$10(k-3)\frac{\ell-1}k$ bins by FF.

The first  $10 \ell$ larger items  are denoted by $ a_{i,p}$ for
$i=1,\ldots ,10$ and $p=1,\ldots ,\ell$, and their sizes are
defined as follows. The item $a_{i,p}$ has size ${\frac 16}+{\eps \over
3^p}-\delta$ for $1\leq i \leq 3$, ${1\over 6}+{\eps \over
3^p}-2\delta$ for $4\leq i \leq 5$, ${1\over 6}-{\eps \over
3^{p+1}} -\delta$ for $6\leq i \leq 7$, and ${1\over 6}-{\eps
\over 3^{p+1}} -2\delta$ for $8\leq i \leq 10$.
Note that all item sizes are in $(\frac 17,\frac 15)$ since the
smallest item has size $\frac 16-\frac{\eps}9-2\delta
>\frac 17$ as $\frac{\eps}9+2\delta < \eps <\frac{1}{120}$, and
the largest item has size $\frac 16+\frac{\eps}3-\delta < \frac
15$.

Every ten items are packed into two bins as follows. The items
$a_{1,p},a_{2,p},a_{3,p},a_{6,p},a_{7,p}$ are packed into one bin
and $a_{4,p},a_{5,p},a_{8,p},a_{9,p},a_{10,p}$ are packed into
another bin. We call these bins $B_{2p-1}$ and $B_{2p}$. The total
size of items in $B_{2p-1}$ is $\frac 56
+\frac{7\eps}{3^{p+1}}-5\delta$. Note that the least loaded bin out
of the first $2\ell$ bins has load of $\frac
56+\frac{\eps}{3^{\ell}}-10\delta > \frac 56+
\frac{\eps}{3^{\ell+1}}$, since $10\delta <
\frac{10\eps}{3^{\ell+4}} < \frac {2\eps}{3^{\ell+1}}$.

The next $10 \ell $ items are denoted by $b_{i,p}$ for $i=1,\ldots
,10$ and $p=1,\ldots ,\ell$.  Their sizes are defined as follows.
The size of $b_{i,p}$ for $1 \leq i \leq 5$ is ${1\over 3}+{\eps
\over 3^{p-1}}-i\delta$ and for  $6 \leq i \leq 10$ it is ${1\over
3}-{\eps \over 3^p}-(i-5)\delta$.

There are $5\ell$ bins are created {from} these items, bins
$B_{2\ell}+5(p-1)+j$, where $1 \leq j \leq 5$ contains items
$b_{j,p}$ and $b_{j+5,p}$. The total size of items in each such
bin is $\frac 23 +\frac{2\eps}{3^p}-2j\delta$. The least loaded bin has a load
of $\frac 23+\frac{2\eps}{3^{\ell}}-10\delta>\frac
23+\frac{2\eps}{3^{\ell}}-\frac{10\eps}{3^{\ell+4}}>\frac
23+\frac{\eps}{3^{\ell}}$.


The last $10 \ell $ items  are denoted by $c_i$ for $i=1,2,\ldots
,10 \ell$, each of these has size of ${1\over 2}+\frac{\delta}2$. These items are packed into the
dedicated bins $B_{7\ell+j}$ for $1 \leq j \leq 10\ell$.

Using the result of \cite{JDUGG74}, the larger items can be packed
into $10\ell+O(1)$ bins. We show that the tiny items can be
combined into these bins, where every such bin receives $k-3$ tiny
items, as there will be three larger items packed into each bin.
For $i=1,\ldots ,5$ and $p=1,\ldots ,\ell$ there is a bin
containing $\{ a_{i,p}, b_{5+i,p}, c_{5(p-1)+i} \}$. For
$i=1,\ldots ,5$ and $p=3,\ldots,\ell$ there is a bin containing
$\{ a_{5+i,p-2}, b_{i,p},c_{5(p+\ell-3)+i}\}$. This gives a total
of $10\ell-10$ bins, each having a level of at most $1-\del$, and
thus $k-3$ tiny items can be added to each bin. All tiny items are
packed, and this leaves thirty unpacked items. The remaining items
are packed into $12$ additional bins: five bins containing $\{
c_{10(\ell-1)+i},b_{i,1}\}$ for $i=1,\ldots ,5$, five bins
containing $\{ c_{10(\ell-1)+5+i},b_{i,2}\}$ for $i=1,\ldots ,5$,
and two bins with five items each, a bin with $\{
a_{6,\ell},a_{7,\ell},a_{8,\ell},a_{9,\ell},a_{10,\ell}\}$, and a
bin with $\{
a_{6,\ell-1},a_{7,\ell-1},a_{8,\ell-1},a_{9,\ell-1},a_{10,\ell-1}\}$.

The number of bins used by FF is $10(k-3)(\ell-1)/k+17\ell$, while
an optimal solution requires at most $10\ell+2$ bins. Thus, the
asymptotic competitive ratio is at least $2.7-\frac 3k$.

Next, we prove upper bounds.  The next two lemmas will be used for
all values of $k \geq 2$.
\begin{claim}\label{onebins}
Every bin of $OPT$ has at most one item of a $1$-bin of FF.
\end{claim}
\begin{proof}
Assume by contradiction that this is not the case, and items $i,j$ of one bin of $OPT$ are packed into $1$-bins by FF. When $j$ arrives, since $s_i+s_j \leq 1$, FF does not open a new bin for $j$, as there is at least one existing bin where it can be packed, a contradiction.
\end{proof}

\begin{claim}\label{veryuseful}
Let $1 \leq j \leq k-1$. Every $j$-bin except for at most one bin has level above $\frac{j}{j+1}$. Moreover, every $j^+$-bin except for at most one bin has level above $\frac{j}{j+1}$.
\end{claim}
\begin{proof}
Assume that there exists a $j$-bin (or $j^+$-bin) whose level is at most $\frac{j}{j+1}$. All further $j^+$-bins (that appear later in the ordering of FF) only have items of sizes above $\frac 1{j+1}$, and each such bin has at least $j$ items, so their levels are above $\frac j{j+1}$.
\end{proof}

We start with a simple proof that the upper bound $2.7-2.4/k$ on the competitive ratio of FF holds in the absolute sense, that is based on the proof of \cite{KSS75}.

\begin{proposition}
The absolute competitive ratio of FF is at most $2.7-2.4/k$.
\end{proposition}
\begin{proof}
Let $L$ be an input, and partition it into two subsequences, $L_1$ that consists of all items that are packed into bins eventually having $k$ items, and $L_2=L \setminus L_1$. By the definition of FF, running it on $L_1$ results in the same bins for these items as in the run on $L$, and the same is true for $L_2$, even if FF is applied without taking the cardinality constraint into account. Let $M_1=FF(L_1)$ and $M_2=FF(L_2)$ be the resulting numbers of bins, where $FF(L)=M_1+M_2$. We will use $OPT(L) \geq \frac {|L|}{k}$ and $|L|-|L_2|=|L_1|=kM_1$. Since the output for $L_2$ is valid without cardinality constraints, we have $FF(L_2) \leq 1.7 OPT(L_2)$.

First, consider the case $M_2 \leq OPT(L)$. Since every bin of FF has at least one item, we have $|L_2|\geq M_2$, and therefore $M_1+M_2 =M_1+M_2/k+(1-1/k)M_2 \leq \frac{|L_1|}k+\frac {|L_2|}k+(1-1/k)M_2=\frac{|L|}k+(1-1/k)M_2 \leq (2-1/k)OPT(L)$.

Otherwise, as the number of $1$-bins is at most $OPT(L)$, and the remaining bins for $L_2$ have at least two items, thus, $|L_2| \geq M_2+(M_2-OPT(L))$, and we get $M_1+M_2 \leq
M_1+2M_2/k+(k-2)M_2/k \leq \frac{|L_1|}k+|L_2|/k+OPT(L)/k+(k-2)M_2/k= |L|/k +OPT(L)/k +(k-2)/k\cdot 1.7OPT(L_2)\leq  OPT(L)+OPT(L)/k+(1.7k-3.4)OPT(L)/k= (2.7-2.4/k)OPT(L)$.
\end{proof}

\begin{claim}
\label{reduc}
For every input $\sigma$ for FF there exists an input $\sigma'$ for FF that contains the same items (possibly in a different order),  $FF(\sigma)=FF(\sigma')$, $OPT(\sigma)=OPT(\sigma')$, and in the output of FF for $\sigma'$, all $k$-bins appear before all bins that are not $k$-bins, and all $1$-bins appear after all $2^+$-bins.
\end{claim}
\begin{proof}
Given $\sigma$, and the output for FF for it, we remove all items of $k$-bins and of $1$-bins of $FF$ from $\sigma$, we append all items of $1$-bins at the end of the input in some order, and we insert all items of $k$ bins in the beginning of the input, in the same order as they appear in $\sigma$. This defines $\sigma'$. Obviously $OPT(\sigma')=OPT(\sigma)$. The items of $k$-bins are packed for $\sigma'$ exactly as they are packed for $\sigma$, as all the items of $k$-bins are presented in the same order. Afterwards, the items of $2^+$-bins are packed exactly as for $\sigma$, since no further item can be packed into a bin already containing $k$ items. Finally, since no two items of $1$-bins can be packed into a bin together, and they cannot join $k$-bins, it remains to show that no such item can join a $2^+$-bin. Let $B$ be a $2^+$-bin, and let $i$ be an item of a $1$-bin $B'$. If $B$ appears earlier than $B'$ in the ordering of FF (applied on $\sigma$), then when item $i$ is presented, it cannot be packed into $B$ (which at that time contains a subset of the items that $B$ receives).  If $B$ appears later than $B'$ in the ordering of FF (applied on $\sigma$), then no item of $B$ can be packed with $i$ into a bin, and obviously $i$ cannot be added to $B$.
\end{proof}

By Claim \ref{reduc}, in what follows we will only analyze inputs where the condition of the claim for $\sigma'$ holds.

\subsection{Analysis of the absolute competitive ratio for the cases $\boldsymbol{k=2,3,4}$}
We start with the simple case $k=2$. A simple
upper bound of $\frac 32$ is achieved by a greedy matching
algorithm, which is a generalization of FF. It is folklore that
this algorithm matches at least half of the edges that an optimal
solution can match and therefore it translates into a $\frac
32$-competitive algorithm for bin packing (where an edge between
two items exists if they can be packed together into a bin). Moreover, for this case the upper bound $2.7-2.4/k$ is equal to $1.5$. For
completeness, and as an introductory
case for analysis using weights, we show how FF can be analyzed using weights for the case $k=2$.  The usage of weights is slightly different from their usage for proving lower bounds. We usually use a weight function $w$, that is applied on sizes of items. Thus, we define $w(a)$ for $a\in (0,1]$, where the variable $a$
denotes the size of an item.  For a set of items $A$ and a set of bins $\mathcal{A}$, let $w(A)$ and
$w(\mathcal{A})$ denote the total weight of all items of $A$ or $\mathcal{A}$.
Furthermore, let $W=w(I)$ be the total weight of all items of the input $I$.
In this kind of analysis,  the weights of bins of the algorithm and of OPT are compared, using the property that for a fixed input, the total weight of items is equal for all algorithms.
An item of an $i$-bin of FF is assigned a weight of $\frac 1i$ (for $i=1,2$). Obviously, any bin of FF has weight $1$, and we analyze the total weight of bins of $OPT$. A bin of $OPT$ cannot have two items of $1$-bins, and therefore its weight cannot exceed $\frac 32$. We find that for any input $L$, the total weight satisfies $FF(L)=W\leq 1.5OPT(L)$.

\subsubsection{The case $\boldsymbol{k=3}$}

In this section we show that the absolute competitive ratio of FF for $k=3$ is exactly $\frac{11}6 <2$ (and that the asymptotic competitive ratio of FF is also equal to this value).

\begin{theorem}
The absolute  approximation ratio of FF for $k=3$ is at most $\frac{11}{6}$.
\end{theorem}
\begin{proof}
Next, let $I$ be an input sequence of items. Recall that it can be assumed without
loss of generality that $3$-bins are positioned in the beginning
of the output, while $1$-bins are positioned in the end of the
output. Thus, the output is sorted by the numbers of items in the
bins. Restricting our attention to the $2$-bins and $1$-bins we
can see that these bins would have been created by running FF only
on the subsequence of the items packed into them, even if the
cardinality constraint is not taken into account. Thus, as in
\cite{BDE}, it can be assumed that no $2$-bin contains two items
that are packed together in an optimal solution, since merging
them into one item would result in the same packing (both for the
application of FF on the original input and for the application of
FF on the items of $2$-bins and $1$-bins). Moreover, if the number
of $1$-bins is $OPT(I)$, then no bin of the optimal solution
contains two items that are packed into $2$-bins (as in
\cite{BDE}). For completeness, we prove this property. Consider
two $2$-bins $B$ and $B'$ (where $B'$ appears later than $B$ in
the ordering). Let $i_1$ and $i_2$ be the items of $B$, and let
$i_3$ be the item of $B'$. Assume that $i_2$ and $i_3$ are packed
into the same bin of $OPT(I)$. Let $i_4$ be the item of that bin
of $OPT(I)$ that is packed into a $1$-bin of $FF(I)$ and $i_5$ is
the item of a $1$-bin of $FF(I)$ that is packed with $i_1$ in
$OPT(I)$. We find $s_{i_3}+s_{i_1}+s_{i_2}>1$ and
$s_{i_4}+s_{i_5}>1$, as $i_3$ was not packed into $B$, and
$i_4,i_5$ are packed into $1$-bins (the item out of $i_4$ and
$i_5$ that arrives later was not packed with the other item out of
these two items). On the other hand, $s_{i_3}+s_{i_2}+s_{i_4} \leq
1$ and $s_{i_1}+s_{i_5} \leq 1$. We have $ 2 <
s_{i_1}+s_{i_2}+s_{i_3}+s_{i_4}+s_{i_5} \leq 2$, a contradiction.

We split the analysis into cases.

\noindent{\bf Case 1.\ } The number of $1$-bins is $OPT(I)$. An optimal solution has at most one such item in each bin, and thus every bin of $OPT$ contains such an item. Additionally it can contain at most one item packed into a $2$-bin by FF.
We define a weight function based on the packing of FF. An item packed into an $i$-bin has weight $\frac 1i$. We find that any bin of the optimal solution has weight of at most $1+\frac 12+\frac 13=\frac {11}{6}$.

\noindent{\bf Case 2.\ } The number of $1$-bins is at most $OPT(I)-1$. In this case there exists at least one bin of the optimal solution that does not contain an item packed into a $1$-bin by FF. We define slightly different weights in this case. An item packed into an $i$-bin by FF, where $i\neq 2$ has weight $\frac 1i$. For the $2$-bins, we define weights as a function of the sizes of items. An item of size above $\frac 12$ (called a big item) has weight $\frac 23$, an item of size in $(\frac 14,\frac 12]$ (called a medium item) has weight $\frac 12$, and an item of size in $(0,\frac 14]$ (called a tiny item) has weight $\frac 13$. Note that there is at most one item packed into a $1$-bin whose size does not exceed $\frac 12$. If such an item exists, then we call it the special item. If its size at most $\frac 14$, then we say that the special item is small, and otherwise it is large. Note that the weight of a bin of $OPT$ that does not have an item of an $1$-bin of FF has a total weight of at most $\frac 32$.

Obviously, the only bins whose total weights can be below $1$ are $2$-bins.
We claim that there exists at most one $2$-bin whose total weight is strictly below $1$, and if there exists a special item and it is small, then no such bin can exist. Consider a $2$-bin that contains a big item. The weight of the big item is $\frac 23$, and the weight of the other item is at least $\frac 13$. Consider a $2$-bin that has level above $\frac 34$. If the bin contains a big item, then we are done. Otherwise, since both the items packed in it have sizes of at most $\frac 12$, none of them can be tiny, and each one of the items has weight $\frac 12$. Assume now that there are two $2$-bins, each of total weight below $1$. Each such bin contains no big items, and at most one medium item. Thus, each such bin contains a tiny item. This contradicts the action of FF as the tiny item that is packed into the bin that appears later in the ordering could be packed into the bin that appears earlier in the ordering.  If there is a special item that is tiny, and there exists a $2$-bin with level at most $\frac 34$, then the special item could be packed there by FF contradicting its action.

Let $W$ denote the total weight. We split the analysis further.

\noindent{\bf Case 2.1\ } There is no special item. We calculate the total weight. There is at most one $2$-bin of weight below $1$, and this bin still has weight of at least $\frac 23$ (as it has two items). Thus, $W \geq (FF(I)-1)+\frac 23=FF(I)-\frac 13$. Every bin of the optimal solution that has an item of a $1$-bin has an item of weight $1$ and size above $\frac 12$. As there is no special item, and it can have only one further item of weight above $\frac 13$, and if it exists, then this item must have size in $(\frac 14,\frac 12]$, and weight $\frac 12$. Thus, the total weight of the bin is at most $\frac{11}{6}$. A bin that only has items of $3$-bins and $2$-bins cannot have an item of weight $1$. It can have at most one item of weight $\frac 23$, in which case the total weight of the two additional items is at most $\frac 12+\frac 13$ (bins with less than three items can only have smaller total weights). If the bin does not have an item of weight $\frac 23$, then the total weight is at most $\frac 32$ as well, since the weight of each item is at most $\frac 12$. We find $W \leq \frac{11}{6}(OPT(I)-1)+\frac 32=\frac{11}{6} OPT(I)-\frac 13$. We found that $FF(I) \leq \frac{11}{6} OPT(I)$ holds in this case.

\noindent{\bf Case 2.2\ } There is a small special item. In this case $W \geq FF(I)$. The total weights of bins of the optimal solution not containing the special item are as computed before. Recall that the special item is an item of a $1$-bin, thus there exists at least one bin whose total weight is at most $\frac 32$. The weight of the bin containing the special item can be larger by $\frac 23$ compared to a bin containing an item of the same size that is not the special item, but not containing an item of a $1$-bin. Therefore,  $W \leq \frac{11}{6}(OPT(I)-2)+2 \cdot \frac 32+\frac 23=\frac{11}{6} OPT(I)$. We found that $FF(I) \leq \frac{11}{6} OPT(I)$ holds in this case as well.

\noindent{\bf Case 2.3\ } There is a large special item.
In this case, if there exists a $2$-bin of FF of level at most $\frac 34$, still its level is above $\frac 12$, as otherwise FF could pack the special item in this bin. At least one of its two items is medium, and the total weight of the bin is at least $\frac 56$. Therefore,  $W \geq FF(I) -\frac 16$. The total weights of bins of the optimal solution not containing the special item are as computed before. As before, the special item is an item of a $1$-bin, thus there exists at least one bin whose total weight is at most $\frac 32$. The weight of the bin containing the special item can be larger by $\frac 12$ compared to a bin containing an item of the same size that is not the special item, but not containing an item of a $1$-bin. Therefore,  $W \leq \frac{11}{6}(OPT(I)-2)+2 \cdot \frac 32+\frac 12=\frac{11}{6} OPT(I)-\frac 16$. We found that $FF(I) \leq \frac{11}{6} OPT(I)$ holds in this case as well.
\end{proof}

\subsubsection{The case $\boldsymbol{k=4}$}
We prove that $FF$ is $2$-competitive in the absolute sense for $k=4$. We define weights as follows. A large item, i.e. any item whose size exceeds $\frac 12$ has weight $1$. A medium item, i.e., an item of size in $(\frac 14,\frac 12]$ has  weight $\frac 12$. A small item, i.e., an item of size at most $\frac 14$ has  weight $\frac 14$. Recall that the total weight of the item is denoted
by $W$.

\begin{lemma}
\label{opt}The weight of any bin of $OPT$ is at most $2$.
\end{lemma}
\begin{proof}
Consider a bin $B$ of $OPT$. Bin $B$ can contain at most one large item.
If $B$ does not contain a large item, then the weight of any item
is at most $\frac 12$, and since $|B|\leq 4$, the
total weight is at most $2$. Suppose now that $B$ contains a large item. Out of the remaining (at most) three items, at most one item can be medium, and the total weight is at most $1+\frac 12+2\cdot \frac 14=2$, and the claim follows.
\end{proof}

\begin{claim}
Every bin that has a large item has total weight of at least $1$.  Every $4$-bin has total weight of at least $1$.
Every $2^+$-bin that does not have any small items has total weight of at least $1$. Every bin that has items of total size above $\frac 34$ has total weight of at least $1$.
\end{claim}
\begin{proof}
The first property holds as the weight of a large item is $1$. The second property holds as the weight of any item is at least $\frac 14$. The third claim holds as two medium items have total weight of $1$. Finally, we prove the last claim. The total size of items of a $2$-bin that does not have a large item but has a small item is at most $\frac 34$. If a $3$-bin only has small items, then the total size of its items is at most $\frac 34$. A $3$-bin that has at most two small items has an item of weight at least $\frac 12$, so its total weight is at least $1$.
\end{proof}

\begin{claim}
Given an input $L$, the total weight of the input items is at least $FF(L)-\frac 34$.
\end{claim}
\begin{proof}
We calculate the total weight based on the packing of FF.
Every dedicated bin, except for possibly one such bin, has a large item, thus there is at most one dedicated bin whose weight is below $1$ (and it is at least $\frac 14$). If all $2^+$-bins have weights of at least $1$, then we are done as all $2^+$-bins, all $4$-bins, and all $1$-bins except for possibly one bin (that has weight of at least $\frac 14$) have total weights of at least $1$.

Otherwise, there must be a $2^+$-bin of level of at most $\frac 34$, that does not have a large item. Consider the first such bin according to the ordering of FF, and call it $B$. This bin must have a small item.
All items that are packed into bins that appear later in the ordering of FF can only be medium or large. Thus, all further $2^+$-bins have total weights of at least $1$, and if there is a $1$-bin that does not have a large item, then it must have a medium item.
Moreover, if $B$ has level of at most $\frac 12$, then all dedicated bins have large items. In this last case, all bins except for $B$ have weights of at least $1$, while $B$ has weight of at least $\frac 12$. Otherwise, if the level of $B$ is in $(\frac 12 , \frac 34]$, then we have the following cases. If $B$ is a $2$-bin, then it has one medium item and one small item, and its weight is $\frac 34$. If $B$ is a $3$-bin, then it has weight of $\frac 34$ as well (as its weight is below $1$, and a $3$-bin has weight of at least $\frac 34$). Thus, all bins except for $B$ and possibly one $1$-bin with a medium item have total weights of at least $1$, while these two bins have weights of at least $\frac 34$ and $\frac 12$, respectively, and the claim follows.
\end{proof}

We have $W \leq 2 \cdot OPT(L)$ and $W \geq FF(L)-\frac 34$. Thus, $FF(L)-2 \cdot OPT(L) \leq \frac 34$, which implies (by integrality) $FF(L) \leq OPT(L)$.


\subsection{The case $\boldsymbol{k=5}$}
We analyzed the cases $k=2,3,4$, and it remains to analyze the asymptotic competitive ratio for
the cases $k\geq 5$, which are more complicated. In particular,
items that are packed in $k$-bins will be treated separately.
These items are called $\alpha$-items, and the weight of every such item will be
equal to $\frac 1k$ in all remaining cases.
Often these items will be analyzed together
with very small items. Items that are not $k$-items will be called
{\it additional} items. Thus, for $k=5$, the weight of any $\alpha$-item is $\frac 15$, and we define weights and types for the additional items as follows. Recall that the variable $a$
denotes the size of an item.

\[
w(a)=
\begin{cases}
1/5 & \text{\ \ \ \ \  if \ \ \ \ \ \ \ \ \ \ \ \ }a\leq1/6\text{, \
\ \ \ in this case the item is  tiny}\\
4/15 & \text{\ \ \ \ \  if \ \ \ }1/6<a\leq1/4\text{, \ \ \ \ \ in this
case the item is  small}\\
7/15 & \text{\ \ \ \ \ if \ \ \ }1/4<a\leq1/3\text{, \ \ \ \ \ in this
case the item is  medium}\\
8/15 & \text{\ \ \ \ \ if \ \ \ }1/3<a\leq1/2\text{, \ \ \ \ \ in this
case the item is  big}\\
1 & \text{\ \ \ \ \ if \ \ \  }1/2<a\leq1\text{, \ \ \ \ \ \  \ \ in this
case the item is  huge}%
\end{cases}
\]

We will show that the weight of any bin of $OPT$ is at most
$32/15$, while the weight of any bin of FF is at
least $1$, except for a constant number of special bins.

\begin{lemma}
For every bin $B$ of $OPT$, $w(B)\leq 32/15$ holds.
\end{lemma}
\begin{proof}
Bin $B$ can contain at most one huge item. Assume
first that $B$ contains a huge item. If it also contains a big
item, then every remaining item is either tiny or an
$\alpha$-item, that is, an item of weight $\frac 15$. The total weight is therefore at most $1+\frac
8{15}+3\cdot \frac 15=\frac{32}{15}$. If $B$ does not contain a
big item, then it can have at most two items that are medium or small, out
of which at most one can be medium, and the remaining items have
weights of $\frac 15$. In this case the total weight is at
most $1+\frac{7}{15}+\frac{4}{15}+2\cdot \frac 15=\frac{32}{15}$.

If $B$ does not contain a huge item, then it can contain at most
three items of sizes above $\frac 14$, out of which at most two
can have sizes above $\frac 13$, and the remaining items have
weights of at most $\frac4{15}$. The total weight is at most
$2\cdot \frac{8}{15}+\frac{7}{15}+2\cdot
\frac{4}{15}=\frac{31}{15}$.
\end{proof}

Now we consider the  bins created by FF.
\begin{lemma}
The total weight of $j$-bins of FF for a given input $L$ is at
least $FF(L)-4$.
\end{lemma}
\begin{proof}
As $k$-bins always have weight $1$, it is left to consider $j$-bins for $1\leq j\leq4$, which contain additional items, and thus their weights are according to $w$. We remove all bins containing total weight
at least $1$ from the considered set of bins, and we will prove
that at most four bins are left. As FF acts in the same way on
subsequences where the complete sets of items of a subset of bins
is given, this will prove the claim. Any bin containing a huge item has
weight of at least $1$, and thus no such bins remain.  Since all $1$-bins except for
at most one bin have huge items, at most one $1$-bin remained,
and if there is such a bin, then it must appear last. Assume by
contradiction that at least five bins have remained, denote the
$\ell$th bin by $B_{\ell}$. The first four bins are $2^+$-bins.

In the next three claims we consider the possible contents of $2^+$-bins that have total
weights below $1$ together with lower bounds on total sizes of
items of such bins.

\begin{claim}
A $2$-bin $B$ such that $w(B)<1$ has items of total size at most $\frac 34$.
\end{claim}
\begin{proof}
Assume by contradiction that the total size of items is above
$\frac 34$. At least one of the items of $B$ must be big, as otherwise the total size is at most $\frac 23$. The second item must be either medium of big, as otherwise the total size is at most $\frac 34$. Thus, the total weight is at least $1$, a contradiction.
\end{proof}

\begin{claim}
Bins $B_1$ and $B_2$ have total sizes of items above $\frac 34$.
\end{claim}
\begin{proof}
Assume by contradiction that the claim does not hold. The further bins cannot have items of sizes in $(0,\frac 14]$ as such an item could be packed into one of $B_1$, $B_2$. A $3^+$-bin with medium and big items has a weight above $1$, thus $B_3$ and $B_4$, that are $2^+$-bins must be $2$-bins. None of them can have a big item, since the total weight of big item and a medium item is $1$. Thus, $B_3$ and $B_4$ have two
medium items each. However, the first medium item of $B_4$ could be
packed into $B_3$, which is a contradiction.
\end{proof}

We find that $B_1$ and $B_2$ are $3^+$-bins.

\begin{claim}
The total size of the items of $B_1$ is at most $\frac 56$.
\end{claim}
\begin{proof}
Assume by contradiction that the total size of items is above
$\frac 56$. We analyze the contents of a $3^+$-bin with items of total size above $\frac 56$ and weight below $1$.

If $B_1$ is a $4$-bin, then it cannot contain an item of size above $\frac 14$, as in such a case the total weight is at least $\frac 7{15}+3\cdot \frac 15>1$.
It cannot have at least two tiny items, since the total size of two tiny items and two small items is at most $\frac 56$. However, the total weight of three small items and another small or tiny item is at least $1$. Thus $B_1$ is a $3$-bin. The total size of three small items is at most $\frac 34$, thus $B_1$ has a medium or big item. It cannot have more than one such item, as in this case the total weight is at least $2\cdot \frac{7}{15}+\frac 15>1$. Since the total size of a medium item and two small items is at most $\frac 56$, we find that $B_1$ has a big item and two items that are small or tiny. Bin $B_1$ has at most one tiny item as the total size of a big item and two tiny items is at most $\frac 56$.  The total weight of a big item, a small item, and another small or tiny item, is at least $1$. We reached a contradiction.
\end{proof}

We find that the further bins do not have tiny items (as a tiny item could be packed into $B_1$). Thus, $B_2$ is a  $3$-bin, as a $4$-bin with no tiny items has a weight of at least
$\frac{16}{15}$, while $w(B_2)<1$.

\begin{claim}
The total size of the items of $B_2$ is below $\frac 34$.
\end{claim}
\begin{proof}
Assume by contradiction that the total size of items is above
$\frac 34$. As $B_2$ is a $3$-bin, it must have at least one item that
is not small, while the total weight of a medium item and two small items is $1$.
\end{proof}

\noindent We have reached a contradiction, and thus the lemma is proved.
\end{proof}

We found that $FF(L) \leq W \leq \frac{32}{15} OPT(L)+4$ for any input $L$.
\begin{theorem}
The asymptotic approximation ratio of FF for $k=5$
is at most $\frac{32}{15}$.
\end{theorem}

\subsection{The cases $\boldsymbol{k=6,7,8}$}
In this case the definitions of the different types of additional
items remain the same, but the weights of such items are defined
differently. The weight of any huge additional item is $1$. Next,
we consider the remaining items, i.e., the additional items with
sizes at most $1/2$. The weight $w(a)$ of any additional item of size $a \leq \frac 12$ consists of three parts. The first part is the ground
weight, the second part is the scaled size, and the third part is
the bonus. Each part is non-negative. The ground weight of any
item of size $a$, is $g(a)=1/k$. This ensures that the weight of
any item (no matter how small it is) is at least $1/k$. The scaled
size of an additional item of size $a \leq 1/2$, is defined by
$s(a)=\frac{2\left(  2k-11\right)  }{3k}a$.

The bonus of an item of size $a$, denoted by $b(a)$ is defined as
follows.
\[
b(a)=%
\begin{cases}
0 & \text{\ \ \ \ if  \ \ \ \ \ \ \ \ \ \ ~ }a\leq1/6\text{ \ \ \ \ (tiny) }\\
\frac{2\left(  2k-11\right)
}{3k}(a-\frac{1}{6})+\frac{10-k}{9k}=\frac{2\left(  2k-11\right)
}{3k}a+\frac{7-k}{3k} & \text{ \ \ \  if
\ \ \ \ }1/6<a\leq1/4\text{ \ \ \ \ (small) }\\
\frac{2\left(  2k-11\right)
}{3k}(a-\frac{1}{4})+\frac{3}{2k}=\frac{2\left(  2k-11\right)
}{3k}a+\frac{10-k}{3k} &
\text{\ \ \ \ if \ \ \ \ }1/4<a\leq1/3\text{ \ \ \ \ (medium)}\\
\frac{2}{k} & \text{\ \ \ \ if \ \ \ \ }1/3<a\leq1/2\text{ \ \ \ \ (big) }%
\end{cases}
\]

Note that $b(a)$ (and therefore also $w(a)$) is a piecewise linear
function. The value of the bonus is zero if $a\leq1/6$, and the
bonus is constant ($2/k$) for $a \in (1/3,1/2]$. It is monotonically non-decreasing for $a \in (0,1/2]$.
The weight of an additional item of size $a \leq 1/2$, is $w(a)=g(a)+s(a)+b(a)$.
The weight function has the discontinuity points, $1/6$, $1/4$, $1/3$, and $1/2$ (this is not exactly the same set of discontinuity points are the weight function of FF and standard bin packing \cite{JDUGG74}).

The bonus of a medium item is at least $\frac 3{2k}$ and at most
$\frac{\frac 23\left( 2k-11\right)
}{3k}+\frac{10-k}{3k}=\frac{k+8}{9k}<\frac{2}k$, and the
bonus of a small item is at least $\frac{10-k}{9k}$ and at most $\frac{\frac 24\left(  2k-11\right)  }{3k}+\frac{7-k}{3k}=\frac{1}{2k}<\frac{3}{2k}$.

\subsubsection{Properties of the weighting and the asymptotic bound}

\begin{lemma}
\label{OptWeight68}
For every bin $B$ of $OPT$, $w(B)\leq \frac{8(k-1)}{3k}$ holds.
\end{lemma}
\begin{proof}
We will assume that all items packed into $B$ are additional items, as an additional item has larger weight than an $\alpha$-item of the same size.

\textbf{Case 1:} $B$ contains no huge item. The bin can contain at
most $k$ items, thus the total ground weight is at most $1$.
Similarly, the total scaled size is at most $\frac{2\left(
2k-11\right)  }{3k}$. Thus it remains to bound ${b}(B)$, it
suffices to show that the total bonus of the items in
the bin is at most ${b}(A)\leq  \frac{8k-8}{3k}-1-\frac{2\left(  2k-11\right)  }%
{3k}=\frac{k+14}{3k}$.

If there are two big items in the bin, there can be at most one
further item with a positive bonus, and $b(A)\leq3\cdot\frac{2}{k} \leq \frac{k+14}{3k}$, for $k \geq 4$.
If there is only one big item in the bin, there can be at most
three further
items having positive bonuses. Then $b(A)\leq\frac{2}{k}+3\cdot\frac{k+8}{9k}=\frac{k+14}{3k}$. Now suppose that any item of $B$ has size at most $1/3$.
There can be at most five items in the bin having positive bonuses,
and there can be at most three medium items among them. Thus the
total bonus is at most $b(A)\leq3\cdot
\frac{k+8}{9k}+2\cdot\frac{1}{2k}=\frac{k+11}{3k}$.

\textbf{Case 2:} $B$ contains a huge item. Recall that the weight of the huge
item is $1$, and its size is bigger than $\frac 12$. There can be at
most $k-1$ further items in the bin, their total ground weight is at most $\frac{k-1}k$, and their total scaled size is at most $\frac{2\left( 2k-11\right)  }{3k}\cdot \frac 12$, thus it suffices
to show that the total bonus of
the further items in the bin is at most $\frac{8k-8}{3k}-1-\frac{k-1}{k}-\frac{2k-11}{3k}=\frac{2}{k}$. The total size of remaining items is below $\frac 12$, thus the bin can contain at most two items with positive bonuses.
Moreover, if $B$ contains only one item with a positive bonus, then this bonus is at most $\frac 2k$, and we are done. Otherwise, if it contains two items of positive bonuses, none of them can be big, and at least one of them is small. If both are small, then their total bonus is at most $\frac 1k$. We are left with the case that $B$ contains items of sizes $a_{1}$ and $a_{2}$ where $\frac 16 < a_{1}\leq \frac 14 <a_{2} \leq \frac 13$. Then applying
$a_{1}+a_{2} < \frac 12$, we get that the total bonus is
\begin{align*}
& \frac{2\left(  2k-11\right)
}{3k}a_{1}+\frac{7-k}{3k}+\frac{2\left(  2k-11\right)  }{3k}a_{2}+\frac{10-k}{3k}\leq
\frac{2\left(  2k-11\right)
}{3k} \cdot \frac 12 +\frac{17-2k}{3k}=\frac{2}{3k}.\end{align*}
\end{proof}

Now, we find a lower bound on the total weight of the bins created by FF for an input $L$ and a given $k\in\{6,7,8\}$. The total weight of $1$-bins is at least their number minus $1$, as all $1$-bins except for possibly one bin have huge items. The total weight of $k$-bins is exactly their number. We will show that for each one of the four sets:  $2$-bins, $3$-bins, $4$-bins, and $5^+$-bins, the total weight of items packed into bins of this set is at least the number of such bins minus $2$ (for $5^+$-bins it is at least their number minus $1$). This will show that $W \geq FF(L)-8$. Since the weight of every bin that contains a huge item is at least $1$, we can restrict the analysis to bins that do not contain such items, and for $2 \leq j\leq k-1$ we will only consider $j$-bins that have no huge items.

\begin{claim}
Every $5^+$-bin of level above $\frac 56$ has weight of at least $1$, and the total weight of $5^+$-bins is at least their number minus $1$.
\end{claim}
\begin{proof}
All $5^+$-bins, except for at most one bin, have levels above $\frac 56$.
Consider a $j$-bin $A$ where $5 \leq j \leq k-1$. The ground weight of its items is $\frac jk$, and their scaled size is at least $\frac{5}{6}\cdot\frac{2\left( 2k-11\right)  }{3k}$. If $j=5$, then at least one item has a positive bonus (otherwise the total size is at most $\frac 56$),
and the weight of the
bin is
${w}(A)={g}(A)+{s}(A)+{b}(A)\geq \frac 5k+\frac{5}{6}\cdot\frac{2\left(
2k-11\right)  }{3k}+\frac{10-k}{9k}=1$, since the value of any positive bonus is at least $\frac{10-k}{9k}$. Otherwise, $k \geq 6$, so the ground weight is at least $\frac 6k$, and we are done since
$\frac 1k \geq \frac{10-k}{9k}$.
Since there is at most one $5^+$-bin whose level is at most $\frac 56$, and all $5^+$-bins with level above $\frac 56$ have weights of at least $1$, we find that the total weight of $5^+$-bins is at least their number minus $1$.
\end{proof}

It remain to consider only the $2$-bins, $3$-bins, and $4$-bins. For all of these cases we consider two
subcases. We will show that if the level of a bin is sufficiently large (above $\frac 34$ for $2$-bins, and above $\frac 56$ otherwise), then the total weight of the bin is at least
$1$.  Then, we will consider $j$-bins of smaller levels for for $j=2,3,4$.

\begin{lemma}
Consider a $2$-bin of level above $\frac{3}{4}$, the weight of the bin is at least $1$.
\end{lemma}
\begin{proof}
The bin must have a big item and another item that is either medium or big.
The ground weight is $\frac 2k$, and the scaled size is at least $\frac{2(2k-11)}{3k}\cdot \frac 34$. The total bonus is at least $\frac{2}{k}+\frac{3}{2k}=\frac{7}{2k}$.
The total weight is therefore at least $\frac 2k+\frac{2k-11}{2k}+\frac{7}{2k}=\frac{4+2k-11+7}{2k}=1$.
\end{proof}

\begin{lemma}
Let $j \in \{3,4\}$. Consider a $j$-bin of level above $\frac{5}{6}$, the weight of the bin is at least $1$.
\end{lemma}
\begin{proof}
The scaled size of the bin is at least $\frac{2(2k-11)}{3k}\cdot\frac 56=\frac{10k-55}{9k}=1+\frac{k-55}{9k}$.
For $k=3$, the bin has ground weight $\frac 3k$, and for $k=4$, the bin has ground weight $\frac 4k$. If the bin has at least two items of sizes above $\frac 14$, then their combined bonuses are at least $\frac 3k$, and the total weight is at least $1+\frac{k-55}{9k}+\frac 3k+\frac 3k=1+\frac{k-1}{9k}>1$. Similarly, if a bin has a big item and at least one other item with a positive bonus, their combined bonuses are at least $\frac 2k+\frac{10-k}{9k}=\frac{28-k}{9k}$, and the total weight is at least $1+\frac{k-55}{9k}+\frac 3k+\frac {28-k}{9k}=1$. The remaining cases are considered separately for $k=3$ and $k=4$.

For $k=3$, a bin that has a big item and two tiny items has level of at most $\frac 56$, and a bin that has a medium item and two small items also has level of at most $\frac 56$, thus there are no additional cases and we are done. For $k=4$, a bin that has a big item has weight of at least $1+\frac{k-55}{9k}+\frac 4k+\frac 2k=1+\frac{k-1}{9k}>1$. We are left with the case where $k=4$, the bin has no big items, and it has at most one medium item. If the bin has one medium item, then (since the size of a medium item and three tiny items is at most $\frac 56$), it must have also one small item and their total bonus is at least $\frac{3}{2k}+\frac{10-k}{9k}=\frac{47-2k}{18k}$, and the bin has weight of at least $1+\frac{k-55}{9k}+\frac 4k+\frac {47-2k}{18k}=1+\frac{1}{2k}>1$.
Finally, if it has no medium items, then it must have at least three small items. These three items must have combined total size of at least $\frac 23$, and their total bonus is at least $\frac{2(2k-11)}{3k}\cdot \frac 23+\frac{3(7-k)}{3k}=\frac{19-k}{9k}$. In this case the weight of the bin is at least $1+\frac{k-55}{9k}+\frac 4k+\frac{19-k}{9k}=1$.
\end{proof}

\begin{lemma}
The total weight of the $2$-bins of levels in $(\frac 23,\frac 34]$ is at least their number minus $1$.
\end{lemma}
\begin{proof}
Consider two consecutive $2$-bins of levels in $(\frac 23,\frac 34]$, $B_i$ and $B_j$ (these bins become consecutive if we remove all other kinds of bins from the list of bins).
We prove that ${g}(B_{i})+{s}(B_{i})+{b}(B_{j})\geq1$. Let the level of $B_{i}$ be $2/3+x$ with some $0 < x\leq1/12$.
There are two items in $B_{j}$, their sizes are above $1/3-x \geq \frac 14$ (since they were not packed into $B_i$), and moreover one of them must be big (and the other one is either medium or big).
Their total bonus is at least $\frac 2k+\frac{2\left(
2k-11\right)  }{3k}\cdot (1/3-x)+ \frac {10-k}{3k}=\frac {26+k-6x(2k-11)}{9k}$.
We get
\begin{align*}
 {g}(B_{i})+{s}(B_{i})+{b}(B_{j})
&  =\frac 2k+\frac{2\left(  2k-11\right)
}{3k}\cdot(2/3+x)+\frac{26+k}{9k}-\frac{2x(2k-11)}{3k}=1.
\end{align*}
The number of pairs $i,j$ that are considered is the number of considered bins minus $1$ and the claim follows.
\end{proof}

Since there is at most one $2$-bin whose level is at most $\frac 23$, and all $2$-bins with level above $\frac 34$ have weights of at least $1$, we find that the total weight of $2$-bins is at least their number minus $2$.

\begin{lemma}
The total weight of the $3$-bins of levels in $(\frac 34,\frac 56]$ is at least their number minus $1$.
\end{lemma}
\begin{proof}
Suppose that $B_{i}$ and $B_{j}$ are two consecutive $3$-bins. We prove that ${g}(B_{i})+{s}(B_{i})+{b}%
(B_{j})\geq1$. Let the level of $B_{i}$ be $3/4+x$ with some
$0 < x\leq 1/12$. Then there are three items in $B_{j}$, of sizes
$a_{1}\geq a_{2}\geq a_{3}$, such that all are bigger than
$1/4-x$, and in particular, all are bigger than $1/6$. At least one of them must
be also bigger than $1/4$, otherwise the level of the bin is at
most $3/4$. We have $g(B_{i})+{s}(B_{i}) =
 \frac{3}{k}+\frac{2(2k-11)}{3k}\cdot (\frac 34+x)=\frac{2k-5}{2k}+\frac{2x(2k-11)}{3k}$. Thus, it is sufficient to show ${b}(B_{j}) \geq \frac{5}{2k}-\frac{2x(2k-11)}{3k}$. This holds if $a_2 > \frac 14$, as the bonus of a medium or big item is at least $\frac{3}{2k}$.
If none of the items is big, using $a_1+a_2+a_3 \geq \frac 34$ we have ${b}(B_{j}) \geq \frac{2\left(
2k-11\right)  }{3k}(a_{1}+a_{2}+a_{3})+2\cdot \frac{7-k}{3k}+\frac{10-k}{3k}\geq \frac{2\left(
2k-11\right)  }{3k}\cdot  \frac 34+ \frac{24-3k}{3k} =\frac{5}{2k}$.

We are left with the case that $a_1>\frac 13$, and $\frac 14-x < a_3 \leq a_2 \leq \frac 14$. The bonus of each small item is at least $\frac{2(1/4-x)(2k-11)}{3k}+\frac{7-k}{3k}$, and ${b}(B_{j}) \geq \frac 2k+2(\frac{2(2k-11)}{3k}(\frac 14-x)+\frac{7-k}{3k})=\frac {3}{k}-2\frac{2x(2k-11)}{3k} =  \frac{5}{2k}-\frac{2x(2k-11)}{3k}+\frac{1}{2k}-\frac{2x(2k-11)}{3k}$, where by using $x \leq \frac {1}{12}$, we get $\frac{1}{2k}-\frac{2x(2k-11)}{3k} \geq \frac{1}{2k}-\frac{2k-11}{18k}=\frac{20-2k}{18k} \geq 0$.
\end{proof}

Since there is at most one $3$-bin whose level is at most $\frac 34$, and all $3$-bins with level above $\frac 56$ have weights of at least $1$, we find that the total weight of $3$-bins is at least their number minus $2$.

\begin{lemma}
The total weight of the $4$-bins of levels in $(\frac 45,\frac 56]$ is at least their number minus $1$.
\end{lemma}
\begin{proof}
Suppose that $B_{i}$ and $B_{j}$ are two consecutive
such $4$-bins. We prove that
${g}(B_{i})+{s}(B_{i})+{b}(B_{j})\geq1$. Let the size of $B_{i}$
be $5/6-x$ with some $0\leq x<1/30$. Then there are four items in
$B_{j}$, of sizes $a_{1}\geq a_{2}\geq a_{3}\geq a_{4}$, all are bigger
than $1/6+x$. If any of them is also bigger than $1/4$, we have
\[
{g}(B_{i})+{s}(B_{i})+{b}(B_{j})\geq \frac 4k+\frac 45\cdot\frac{2\left(
2k-11\right) }{3k}+\frac{3}{2k}=\frac{32k-11}{30k}\geq1,
\]
since $k\geq6$ and we are done. Otherwise all four items are
small. Note that the three biggest items among them have total
size at least $3/4\cdot4/5=3/5$, so the total size of all four of them is at least $\frac 35+\frac 16+x=\frac{23}{30}+x$
and the total bonus is at least $\frac{2(2k-11)}{3k}(\frac{23}{30}+x)+\frac{4(7-k)}{3k}$.
Thus we have
\begin{align*}
&  {g}(B_{i})+{s}(B_{i})+{b}(B_{j}) \geq \frac 4k+\frac{2\left(  2k-11\right)
}{3k}\cdot\left(\frac 56-x\right)+\frac{2(2k-11)}{3k}\left(\frac{23}{30}+x\right)+\frac{4(7-k)}{3k}\\
& =\frac {40-4k}{3k}+\frac{2(2k-11)}{3k}\left(\frac 56-x+\frac{23}{30}+x\right)=\frac {40-4k}{3k}+\frac{2(2k-11)}{3k}\cdot \frac{8}{5}  = \frac{4k+8}{5k} \geq 1,
\end{align*}
since $k\leq8$. \end{proof}

Since there is at most one $4$-bin whose level is at most $\frac 45$, and all $4$-bins with level above $\frac 56$ have weights of at least $1$, we find that the total weight of $4$-bins is at least their number minus $2$.

We proved $FF(L)-8 \leq W \leq (8/3-8/(3k))OPT(L)$.
\begin{theorem}
The asymptotic approximation ratio of $FF$\ for any $6\leq k\leq8$
is at most $8/3-8/(3k)$.
\end{theorem}


\subsection{The case $\boldsymbol{k=9}$}
We consider this case separately, as the proof for smaller and larger $k$ fails in this case. We combine methods from all other proofs here. The weighting function is similar to that is used in the previous section, in the sense that it has discontinuity
points at $1/6$ and $1/3$. It is also similar to the weighting used in the next section as the intervals for small and medium
sizes are divided to two parts. In this section we introduce a new method that was not used in the previous cases. We distinguish the bins of $OPT$
according to the number of additional items packed into them. Since $\alpha$-items always have weights of $\frac 1k$, bins that contain $k$ such items will still have weights of exactly $1$. Thus, in the analysis we can assume that there is at least one additional item in each bin of $OPT$.

\noindent \textbf{Case a.} Consider bins of $OPT$ containing one or two
additional items (and the remaining items are $\alpha$-items). Such bins are called
$\gamma$-bins, and the additional items packed into such bins (in $OPT$) are called $\gamma$-items. The largest $\gamma$-item of such a bin is called a $\gamma_{1}$-item (breaking ties arbitrarily). Any $\gamma_1$-item has weight $1$. If the bin contains another $\gamma$-item, this item is called a $\gamma_2$-item, and its weight is defined to be $\frac{16}{27}$.

\noindent \textbf{Case b.} Consider the other bins of $OPT$ (containing
\textit{at least} three additional items). Each such bin has at
most six $\alpha$-items, we call it a $\phi$-bin, and its
items that are not $\alpha$-items are called $\phi$-items. The
weighting function of the $\phi$-items is more complicated. The
weight of any huge $\phi$-item (i.e. a $\phi$-item with size
strictly above $1/2$) is exactly $1$.  The weight of a $\phi$-item
of size $a\leq1/2$ is $w(a)=s(a)+b(a)$, where
$s(a)=\frac{32}{27}a$ is called the scaled size, and $b(a)$ is the
bonus of the item. Note that there is no ground weight in this
case. Below we give the bonus function of the $\phi$-items of
sizes no larger than $1/2$. The functions $b(a)$ and $w(a)$ are piecewise
linear, and breakpoints where it is continuous are $1/5$ and $3/10$.

\[
b(a)=%
\begin{cases}
0 & \text{if \ \ \ \ \ \ \ \ \ \  \ \ \  \ }a\leq1/6\text{\ \ \ \ \ (tiny) }\\
\frac{32}{27}(a-\frac{1}{6})+\frac{1}{81} = \frac{32}{27}a-\frac{5}{27} & \text{if  \ \ \ \ \ }1/6<a\leq1/5\text{  \ \ \ \ (very small)}\\
-\frac{28}{27}(a-\frac{1}{5})+\frac{7}{135} = - \frac{28}{27}a+\frac{7}{27}& \text{if \ \ \ \ \ }1/5 < a\leq
1/4\text{\ \ \ \ \ (larger small)}\\
-\frac{28}{27}(a-\frac{1}{4})+\frac{1}{9} = -\frac{28}{27}a+\frac{10}{27} & \text{if \ \ \ \ \ }1/4<a\leq3/10\text{\ \ \ \ (smaller medium)}\\
\frac{32}{27}(a-\frac{3}{10})+\frac{8}{135}= \frac{32}{27}a-\frac 8{27}& \text{if \ \ \ \ \ }3/10 <
a\leq1/3\text{\ \ \ \ (larger medium)}\\
\frac{1}{9} & \text{if \ \ \ \ \ }1/3<a\leq1/2\text{\ \ \ \ \ (big)}%
\end{cases}
\]

The value of the bonus is zero if $a\leq1/6$ and it is constant ($\frac 19$) between
$1/3$ and $1/2$. The bonus function is not continuous at the points $1/6$, $1/4$,
and $1/3$. The bonus function is monotonically increasing in $(1/6,1/5)$ and in $(3/10,1/3)$. It is monotonically decreasing in $(1/5,1/4)$ and in $(1/4,3/10)$ (which is less typical for weight functions).
Nevertheless, the weight function remains monotonically increasing for the whole
interval $0 < a\leq1/2$, and the value of the bonus is nonnegative for the
whole interval.

Next, we state several additional properties of the bonus
function. For small items (very small and larger small items,
i.e., items of sizes in $(1/6,1/4]$), the maximum value of the
bonus for very small items is given for $a=1/5$, and the bonus at
this point is $7/135$. The bonus of very small items is at least
$\frac 1{81}$, and the smallest bonus of larger small items is
zero. For smaller medium items, the bonus decreases from $\frac
19$ to $\frac 8{135}$, and for larger medium items, the bonus
increases from $\frac 8{135}$ to $\frac 8{81}$. The weight of a big
$\phi$-item is at least $41/81$, the weight of a
$\phi$-item with size more than $1/4$ is at least $11/27$, and
for any $\phi$-item with size $0 < a\leq1$, and for any
$\gamma_{2}$-item, the next inequality holds: $w(a)\geq
\frac{32}{27} a > \frac 76 a$. This is true since bonuses of
$\phi$-items are non-negative, and since the size of any
$\gamma_2$-item  is at most $\frac 12$ (as the bin of $OPT$ that contains it has
another item of at least the size of the $\gamma_2$-item) while
its weight is $\frac{16}{27}$.

\subsubsection{Properties of the weighting and the asymptotic bound}

\begin{lemma}
For every bin $B$ of $OPT$, $w(B)\leq8/3-8/27=2+\frac {10}{27}=\frac{64}{27}$ holds.
\end{lemma}
\begin{proof}
Consider the case that  $B$ is a $\gamma$-bin. In this case $B$ has one item of weight $1$, possibly an item of weight $\frac{16}{27}$, and each remaining item is an $\alpha$-item and has weight $\frac19$. Thus, the total weight is at most $1+\frac{16}{27}+\frac79=\frac{64}{27}$.
It remains to consider the case that $B$ is a $\phi$-bin. It contains at most six $\alpha$-items, of total weight at most $6/9$. Thus it suffices to show that the total weight of the $\phi$-items is at most $\frac{64}{27}-\frac 69=\frac{46}{27}$.

First, assume that $B$ contains no huge item. The total scaled size of the
$\phi$-items is at most $32/27$. It suffices to show that the total bonus of $\phi$-items
the bin is at most $46/27-32/27=14/27$. Since the bonus is zero if the size of
the item is at most $1/6$, it follows that at most five items can have
positive bonuses. Moreover at most three items can have size above $1/4$, and the bonus of each such item is at most $\frac 19$, and the
bonus of any other item is at most $7/135$. Thus the total bonus of the bin is
at most $3\cdot \frac 1 9+2\cdot \frac 7{135}=59/135<14/27$.

Next, assume that $B$ contains a huge item. The weight of a huge item is exactly $1$. We will show that if there are six $\alpha$-items, then the total weight of the further additional items of $B$ is at most $46/27-1=19/27$, and consider also the case that the number of $\alpha$ items is smaller.
Since the total size of remaining additional items is below $\frac 12$, their scaled size is at most $\frac {16}{27}$, and  it
suffices to show that their total bonus is at most $19/27-16/27=1/9$. Since only items of size above $\frac 16$ have positive bonuses, there can be at most two further items in the bin having positive bonuses. If there is only one further item having positive
bonus, we are done, since no bonus is above $\frac 19$. If there are two items with bonuses, but there are at most five $\alpha$-items, then the total weight of $\alpha$-items is at most $\frac 59$, and we are done as well.
Thus, it is left to consider the case where there are two further
$\phi$-items in the bin both having positive bonuses, and there are no other $\phi$-items packed into $B$ except for the huge item and these two items. Let their sizes be denoted as
$a_{1}$ and $a_{2}$, where $1/6<a_{1}\leq a_{2}$, and thus $a_2<1/3$ as $a_{1}+a_{2}<1/2$.
The claim holds if $a_{2}\leq1/4$, since then the total bonus is at most
$2\cdot7/135=\frac{14}{135}<\frac{1}{9}$.
Thus the only remaining case is  where the item of size $a_{1}$ is small and the item of size $a_{2}$ is a medium item. We will show $w(a_1)+w(a_2) \leq \frac{19}{27}$. There are three cases, as $a_1> \frac 15$ and $a_2 > 0.3$ cannot hold simultaneously. In all cases $s(a_1)+s(a_2) = \frac{32}{27}(a_1+a_2)$ and $a_1+a_2<\frac 12$.

If $a_{1} \leq \frac 1{5}$ and $a_{2} \leq \frac 3{10}$, $w(a_1)+w(a_2)=\frac{32}{27}(a_1+a_2)+\frac{32}{27}a_1-\frac{28}{27}a_2+\frac 5{27}=\frac{64}{27}a_1+\frac{4}{27}a_2+\frac 5{27}\leq \frac{64}{27}a_1+\frac{4}{27}(\frac 12-a_1)+\frac 5{27}=\frac{60}{27}a_1+\frac{7}{27}\leq \frac{19}{27}$.
If $a_{1} > \frac 1{5}$ and $a_{2} \leq \frac 3{10}$, $w(a_1)+w(a_2)=\frac{32}{27}(a_1+a_2)-\frac{28}{27}(a_1+a_2)+\frac {17}{27} = \frac{4}{27}(a_1+a_2)+\frac {17}{27}\leq \frac{19}{27}$.
If $a_{1} \leq \frac 1{5}$ and $a_{2} > \frac 3{10}$, $w(a_1)+w(a_2)=\frac{32}{27}(a_1+a_2)+\frac{32}{27}(a_1+a_2)-\frac {13}{27}=\frac{64}{27}(a_1+a_2)-\frac {13}{27}\leq  \frac{19}{27}$.
\end{proof}

Now, we will analyze the total weight of the bins of FF. Once again, we split
the analysis according to the number of items in these bins. The $9$-bins that have weight $1$. Moreover, every item of
size above $\frac 12$ packed into any bin of FF that is not a
$9$-bin is either a huge $\phi$-item, or it is a $\gamma$-item, in
which case this must be the largest item of its bin of $OPT$, i.e.,
it is a $\gamma_1$-item, and its weight is $1$. Thus, we neglect
all bins containing items of size above $\frac 12$ from the
analysis. At most one $1$-bin is left, and we neglect that bin (if
it exists) as well. In what follows we analyze $2^+$-bins of FF
that contain items of sizes in $(0,\frac 12]$. These bins only
contain $\phi$-items, and $\gamma_2$-items.

\begin{lemma}
The weight of any bin with level above $\frac 67$ is at least $1$.
There is at most one $6^+$-bin whose weight is below $1$.
\end{lemma}
\begin{proof}
For any $\phi$-item or for any $\gamma_2$-item, the weight of the
item is at least $\frac{32}{27}$ times the size of the item. Since
except for at most one bin, the level of $6^+$-bin is at most
$\frac 67$, the weights of these bins, except for at most one bin, are at least $1$.
\end{proof}

In the following we concentrate on the $2$-bins, $3$-bins,
$4$-bins and $5$-bins. We start
with analyzing  bins containing a $\gamma_2$-item. Note that a
bin that has two $\gamma_2$-items has weight above $1$, and thus
we consider bins that contain one $\gamma_2$-item and the
remaining items are $\phi$-items.

\begin{lemma}
The weight of any $2$-bin that has a $\gamma_{2}$-item is at least
$1$, except for at most one such bin.
The weight of any $3$-bin that has a $\gamma_{2}$-item is at least
$1$, except for at most two such bins.
The weight of any bin that is a $4$-bin or $5$-bin and has a
$\gamma_{2}$-item is at least $1$, except for at most one such
bin.
\end{lemma}
\begin{proof}
Any bin that has a $\gamma_2$-item and a $\phi$-item of size above
$\frac 14$ has total weight of at least $1$, since any $\phi$-item
of size above $\frac 14$ has weight at least $\frac{11}{27}$, and
each $\gamma_{2}$-item has weight $\frac{16}{27}$. Moreover, since
the weight of any item is at least $\frac{32}{27}$ times its size,
if the level of the bin is at least $\frac {27}{32}$, then the
total weight is at least $1$. In each one of the cases we assume
that there exist two bins, called $B_{i}$ and $B_{j}$, were $B_i$
appears earlier than $B_j$ in the ordering of FF, each having a
$\gamma_{2}$-item, and each one of these two bins has weight below
$1$ (and thus a level below $\frac{27}{32}$), all its
$\phi$-items have sizes of at most $\frac 14$, and the sizes of items of $B_j$ are above $\frac
5{32}$ (since none of them could not be packed into $B_i$).

Bins $B_i,B_j$ cannot be $2$-bins; in such a case $B_i$ has one $\gamma_2$-item (of size at most $\frac 12$) and one $\phi$-item (of size at most $\frac 14$, so its level is at most $\frac 14$, and in this case
$B_j$ cannot have an item of size at most $\frac 14$. Thus, there cannot be a pair $i,j$ that are both $2$-bins.

For $4$-bins and $5$-bins, $B_j$ has at least three $\phi$-items of sizes
at least $\frac 5{32}$, whose total weight is at least $3 \cdot
\frac{32}{27}\cdot \frac{5}{32}=\frac{5}{9}$, so the weight of
$B_j$ is above $1$, a contradiction. Thus, there cannot be a pair $i,j$, each of which containing four or five items.

For $3$-bins, we assume that there is a third $3$-bin $B_r$ that appears after $B_j$ in the ordering, where $B_r$ has a $\gamma_2$-item and weight below $1$ as well.
We split the analysis to the cases where $B_j$ has
level above $\frac 56$, and the case that it does not. If it has
level above $\frac 56$, then the total size of its $\phi$-items is
above $\frac 13$ (as the $\gamma_2$-item has size of at most $\frac 12$), and at least one of them has size above $\frac 16$. If at least one of the two $\phi$-items has size above $\frac 15$, then
their total size is above $\frac5{32}+\frac 15=\frac{57}{160}$,
and the weight of all items is at least
$\frac{16}{27}+\frac{32}{27}\cdot \frac{57}{160} >1$. Otherwise,
there is an item that has a bonus of at least $\frac 1{81}$, and
the total weight is at least $\frac{16}{27}+\frac{32}{27}\cdot
\frac 13 +\frac 1{81}=1$. Therefore, we find that $B_j$
has a level of at most $\frac 56$. However, in this case
the two $\phi$-items of $B_r$ have sizes above $\frac 16$. The
weight of such an item is at least $\frac{32}{27}\cdot \frac 16
+\frac 1{81}=\frac{17}{81}$, and the total weight of $B_r$ is at
least $\frac{16}{27}+2\cdot \frac{17}{81} >1$, a contradiction.
\end{proof}

We are left with bins containing only $\phi$-items of sizes at most $\frac 12$, that are $2$-bins, $3$-bins, $4$-bins, and $5$-bins. Before analyzing their total weights we discuss some properties of $\phi$-items.

\begin{lemma}
\label{aa2}
Consider two $\phi$-items of sizes $a_{1}\leq a_{2}\leq1/2$.
If $1 \geq a_{1}+a_{2}>1-a_{1}$ holds, then the total weight of the two items is at
least $1$.
\end{lemma}
\begin{proof}
We have $a_1 > \frac{1-a_2}2 \geq \frac 14$. If $a_1 > \frac 13$, then both items are big, and since the weight of any big item is at least $\frac{41}{81}>1/2$, the claim holds in this case. Next, we consider the case $1/4<a_{1}\leq1/3$, where $a_{2}>1-2a_{1}\geq1/3$, thus the item of size $a_{2}$ is big.

If $a_{1}$ is smaller medium, then let $a_{1}=1/4+x$ for some $0<x\leq1/20$, and $a_1+a_2>1-a_1=\frac 34-x$. The
total weight of the items is
\[
\frac{32}{27}(a_{1}+a_{2})+b(a_{1})+b(a_{2})\geq\frac{32}{27}(3/4-x)-\frac
{28}{27}x+\frac{1}{9}+\frac{1}{9}=\frac{10}{9}-\frac{20}{9}x\geq1\text{.}%
\]

If $a_{1}$ is larger medium, then let $a_{1}=3/10+x$ for some $0 < x\leq1/30$, and $a_1+a_2>1-a_1=\frac 7{10}-x$.
We get
\[
\frac{32}{27}(a_{1}+a_{2})+b(a_{1})+b(a_{2})\geq\frac{32}{27}(7/10-x)+\frac
{32}{27}x+\frac{8}{135}+\frac{1}{9}=1\text{.}%
\]
\end{proof}

\begin{lemma}
\label{aa3}
Consider three $\phi$-items of sizes  $a_{1}\leq
a_{2}\leq a_{3}\leq1/2$. If $1 \geq a_{1}+a_{2}+a_{3}>1-a_{1}$ holds, then the total
weight of the three items is at least $1$.
\end{lemma}
\begin{proof}
We have $\frac 32(a_1+a_2) \geq 2a_1+a_2>1-a_3 \geq \frac 12$, thus $a_1+a_2 > \frac 13$.

If $a_1 > \frac 14$, then the claim holds since the weight of an
item with size above $1/4$ is at least $11/27$ (so the total weight is at least $\frac{11}9$). In what follows we assume that
$a_{1}\leq1/4$, and thus $a_1+a_2+a_3 > 1-a_1 \geq \frac 34$. If the largest item is big, then its bonus is $\frac 19$, and the total weight of the three items is at least $\frac{32}{27}(a_1+a_2+a_3)+\frac 19 \geq \frac{32}{27}\cdot\frac 34 +\frac 19 =1$.
If the largest item is not big, i.e., $a_3 \leq \frac 13$, then using $\frac 23 \geq 2a_3 \geq a_2+a_3  > 1-2a_1 \geq \frac 12$ we get $a_1 > \frac 16$ and $a_3>\frac 14$, thus the smallest item is small, and the largest item is medium. If there are two medium items, then the bonus of each one of them is at least $\frac 8{135}$, and the total weight is at least $\frac{32}{27}\cdot \frac 34+\frac{16}{135}=\frac{136}{135} >1$. We are left with the case where the smallest two items are small and the largest item is medium. We find that $2a_1 > 1-a_2-a_3 \geq 1-\frac 14-\frac 13= \frac 5{12}$, thus the smallest item is larger small, and so is the item of size $a_2$.

If the largest item is smaller medium, we have a total weight of $\frac{32}{27}(a_1+a_2+a_3)-\frac{28}{27}(a_1+a_2+a_3)+\frac{24}{27} \geq \frac 4{27} \cdot \frac 34+\frac{24}{27}=1$. If the largest item is larger medium, we have $a_1+a_2  > \frac 23(1-a_3)$.
The total weight is at least $\frac{32}{27}(a_1+a_2+a_3)-\frac{28}{27}(a_1+a_2)+\frac{14}{27}+\frac{32}{27}a_3-\frac 8{27}=\frac{4}{27}(a_1+a_2)+\frac{64}{27}a_3+\frac 2{9} \geq
\frac{4}{27}\cdot \frac 23(1-a_3)+\frac{64}{27}a_3+\frac 2{9}=\frac{184}{81}a_3+\frac{26}{81} \geq \frac{184}{81} \cdot \frac 3{10}+\frac {26}{81}=\frac{406}{405}>1$.
\end{proof}

\begin{lemma}
\label{aa4} Consider four $\phi$-items of sizes $a_{1}\leq a_{2}\leq a_{3}\leq a_{4}\leq1/2$. If $1 \geq a_{1}+a_{2}+a_{3}
+a_{4}>1-a_{1}$ holds, then the total weight of the four items is at least $1$.
\end{lemma}
\begin{proof}
First, consider the case $a_1 \leq \frac 16$. In this case $a_{1}+a_{2}+a_{3}+a_{4} > 1-a_1\geq \frac 56$, and if at least one item is very small or medium, then there is at least one item with a bonus of at least $\frac 1{81}$, and the total weight is at least $\frac{32}{27}\cdot \frac 56+\frac 1{81}=1$. Otherwise, all items are larger small and tiny. At least three items must be larger small as $2a_1+a_2+a_3+a_4>1$, which is impossible in the case where $a_2 \leq \frac 16$ and $a_4 \leq \frac 14$. The total weight if all four items are larger small is at least $\frac{4}{27}\cdot (a_1+a_2+a_3+a_4)+\frac{28}{27}>1$. Otherwise, the total weight is at least  $\frac{32}{27}a_1+ \frac{4}{27}\cdot (a_2+a_3+a_4)+\frac{21}{27}>\frac{32}{27}a_1+ \frac{4}{27}\cdot (1-2a_1)+\frac{21}{27}=\frac{8}{9}a_1+\frac{25}{27}$. Since the three largest items are larger small, $a_2+a_3+a_4 \leq \frac 34$, and $2a_1 > 1-(a_2+a_3+a_4) \geq \frac 14$, so $a_1 > \frac 18$. We find that the total weight is at least $\frac{28}{27}>1$.

Next, consider the case $a_{1}>1/6$. It follows that the total weight of the three smallest
items is at least $\frac{32}{27}\cdot \frac 36=\frac{16}{27}$. If the biggest item
is bigger than $1/4$, the total weight is at least $\frac{16}{27}+\frac{11}{27}=1$.
Then it is left to consider only the case where $1/6<a_{1}\leq a_{2}\leq a_{3}\leq
a_{4}\leq1/4$, i.e., all items are small.

Note that the weight of a larger small item is at least $\frac{32}{27}
\cdot \frac 15+\frac{7}{135}=\frac{13}{45}$. If all four items are larger small, then their total weight is above $1$. Otherwise, the smallest item is very small. The total size of the items is above $1-a_1$, and the bonus of the smallest item is $\frac{32}{27}a_1-\frac 5{27}$. The total weight of the four items is at least $\frac{32}{27}(1-a_1)+\frac{32}{27}a_1-\frac 5{27}=1$.
\end{proof}

\begin{lemma}
\label{aa5}
Consider five $\phi$-items of sizes $a_{1}\leq a_{2}\leq a_{3}\leq a_{4}\leq a_5 \leq 1/2$. If $1 \geq a_{1}+a_{2}+a_{3}
+a_{4}+a_5 >1-a_{1}$ holds, then the total weight of the five items is at least $1$.
\end{lemma}
\begin{proof}
If $a_1 \leq \frac 16$, then  $a_{1}+a_{2}+a_{3}
+a_{4}+a_5 > 1-a_1 \geq \frac 56$ holds. Otherwise, $a_{1}+a_{2}+a_{3}
+a_{4}+a_5 \geq 5a_1 >\frac 56$ holds too. If at least one item is not larger small, then its bonus is at least $\frac 1{81}$ and the total weight is at least $\frac{32}{27} \cdot \frac 56+\frac 1{81}=1$. If all items are larger small, then their total size is above $1$, contradicting the assumption.
\end{proof}

\begin{lemma}
\label{9last}The total weight of the $2$-bins, $3$-bins, $4$-bins, and $5$-bins, containing
$\phi$-bins is at least their number minus $1$.
\end{lemma}
\begin{proof}
Consider bins of FF whose numbers of items is in $[2,5]$, that contain only $\phi$-items, and their weights are below $1$. Obviously these bins have no huge items. If the level of a given bin is bigger than $1$ minus the size of the smallest item
in the bin, then the weight of the bin is at least $1$ by the previous lemmas. Thus, we only consider bins that do not satisfy this property. If there is at most one bin to consider, then we are done. Otherwise, in the list of remaining bins, consider two consecutive bins $B_i$ and $B_j$ (such that $B_j$ appears after $B_i$ in the ordering). Let $i_1$ denote the smallest item of $B_i$ and $j_1$ the smallest item of $B_j$ (breaking ties in favor of items of smaller indices). Let $S=s(B_i)$.
We will show that the total weight of the items of $B_i$ excluding $i_1$ together with the weight of $j_1$ is at least $1$. Applying this property to every such consecutive pair of bins will show that the total weight is at least the number of bins in the list of remaining bins minus $1$. If $S-s_{i_1}+s_{j_1}>1$, then their total weight is above $\frac{32}{27}>1$. Otherwise, we have the following properties. First, $s_{j_1}> 1-S$ since $j_1$ was not packed into $B_i$. Additionally, by assumption, $S \leq  1-s_{i_1}$. Therefore $s_{j_1}>s_{i_1}$. Finally, $(S-s_{i_1})+2s_{j_1}>S+s_{j_1}>1$. Thus, the set of items of $B_i$ together with $j_1$ and excluding $i_1$ satisfies the condition of one of Lemmas \ref{aa2},\ref{aa3},\ref{aa4},\ref{aa5} (the one where the considered number of items is equal to that of this set (which is equal to the number of items of $B_i$ and therefore it is in $\{2,3,4,5\}$), and the total weight of this set is at least $1$.
\end{proof}

We proved $FF(L)-7 \leq W \leq (64/27))OPT(L)$.
\begin{theorem}
The asymptotic approximation ratio of $FF$\ for $k=9$
is at most $64/27$.
\end{theorem}

\subsection{The case $\boldsymbol{k\geq10}$}

The case $k\geq10$ is studied similarly to previous cases. In this
case we also distinguish the definitions of weights based on the  bins of $OPT$ according to the number of
additional items packed into these bins. The weight of an $\alpha$-item remains $\frac 1k$.

\noindent \textbf{Case a.} Consider bins of $OPT$ containing one or two
additional items (and the remaining items are $\alpha$-items).
Such bins are called $\gamma$-bins again, and the additional items
in the bin are called $\gamma$-items. The largest $\gamma$-item of
a bin is called a $\gamma_{1}$-item (breaking ties arbitrarily).
Any $\gamma_1$-item has weight $1$. If the bin contains another
$\gamma$-item, this item is called a $\gamma_2$-item. If
$10\leq k\leq 19$, then the weight of the $\gamma_2$-item is
$\frac {7}{10}- \frac 1k$, and otherwise (if $k\geq20$),
then its weight is $\frac{13}{20}=0.65$.
The smallest weight of a $\gamma_2$-item is $0.6$, and its size is
at most $\frac 12$, therefore the ratio between the weight of such
an item and its size is at least $1.2$.

\noindent \textbf{Case b.} Consider the remaining bins of $OPT$ (containing
at least three additional items). Each such bin has at
most $k-3$ $\alpha$-items, and we call it a $\phi$-bin. The items packed into $\phi$-bins that are not $\alpha$-items are called $\phi$-items. The
weighting function of the $\phi$-items is again more complicated.
The weight of any huge $\phi$-item (i.e., a $\phi$-item with size
strictly above $1/2$) is exactly $1$. The weight of a $\phi$-item
of size $a\leq1/2$ is $w(a)=s(a)+b(a)$, where $s(a)=\frac{6}{5}a$
is the scaled size, and $b(a)$ is the bonus of the item. Below we
give the bonus function of the $\phi$-items of sizes no larger
than $1/2$.

For $k\geq20$, the classical weighting function of  FF \cite{JDUGG74} is appropriate, in this case the bonus function is defined as follows.

\[
b(a)=%
\begin{cases}
0 & \text{ \ \ if \ \ \ \ \ \ \ \ \ }a\leq1/6\text{ \ \ \ \ (tiny) }\\
\frac{3}{5}(a-\frac{1}{6})=0.6a-0.1 & \text{ \ \ if }1/6 < a\leq1/3\text{\ \ \ \ \ (small or medium)}\\
0.1 & \text{\ \ if }1/3 <  a\leq1/2\text{ \ \ \ \ (big)}
\end{cases}
\]

The weight function in this case is continuous in the interval $(0,\frac 12)$. The bonus is piecewise linear (and so is the weight function). In the interval $(\frac 16,\frac 13)$, the bonus increases from $0$ to $0.1$.

For $10 \leq k \leq 19$, we use additional modifications to the classic weight function. In some of the cases the bonus function is still equal to
the one in the classic analysis. More
specifically, these are the cases where the size is in $(1/6,1/5]$
and $(3/10,1/3]$. In these intervals the slope of the weight
function is $1.8$, i.e, the slope of the bonus function is $0.6$.
The bonus function and the weight function are piecewise linear, and continuous in $(0,\frac 14)$ and $(\frac 14,\frac 12)$. The partition into types is as in the case $k=9$.

\[
b(a)=%
\begin{cases}
0 & \text{if \ \ \ \ \ \ \ \ \ \ }a\leq1/6\\
\frac{3}{5}(a-\frac{1}{6})=0.6a -0.1 & \text{if }1/6 < a\leq1/5\\
(1.6-20/k)(a-\frac{1}{5})+\frac{1}{50}=(1.6-20/k)a-0.3+\frac 4k & \text{if }1/5 <
a\leq1/4\\ (1.6-20/k)(a-\frac{1}{4})+ \frac 1k = (1.6-20/k)a-0.4+\frac 6k & \text{if
}1/4<a\leq3/10\\
\frac{3}{5}(a-\frac{3}{10})+\frac{8}{100} =0.6a-0.1& \text{if }3/10 < a\leq1/3\\
0.1 & \text{if }1/3 < a\leq1/2\end{cases}
\]

This bonus function is monotonically
non-decreasing for $k \geq 13$, but not in the cases $k=10,11,12$, whereas the resulting weight function is monotonically increasing for $10 \leq k \leq 19$. The value of the bonus is zero for $a\leq1/6$ and it is $0.1$
between $1/3$ and $1/2$.  We have $b(\frac
15)=0.02$ (and $w(\frac 15)=0.26$), $b(\frac 14)=0.1-\frac 1k$, thus for $a \in (\frac 15,
\frac 14]$ the bonus is in $[0.1-\frac 1k,0.02)$ for $k=10,11,12$
and in $(0.02,0.1-\frac 1k]$ for $13 \leq k \leq 19$. For $a \in
(\frac 14,0.3]$ the bonus is in $[0.08,\frac 1k)$ for $k=10,11,12$
and in $(\frac 1k,0.08]$ for $13 \leq k \leq 19$ (we have $w(0.3)=0.44$).

For $k \geq 10$, since the bonus function is non-negative, for any $\phi$-item of size $0\leq a\leq \frac 12$,
$w(a)\geq \frac 65 a$ holds. The bonus of every item of size in $(0,\frac
12]$ is in $[0,0.1]$. The weight of a big item is at least $0.3$. The weight of a medium item is at least $0.3+\frac 1k>0.35$ for $k \leq 19$ and at least $0.35$ for $k \geq 20$.

\subsubsection{Properties of the weighting and the asymptotic bound}
\begin{lemma}
For every bin $B$ of $OPT$, $w(B)\leq 2.7- 3/k$ holds.
\end{lemma}
\begin{proof}
The claim holds for bins having only additional items. For a
$\gamma$-bin, if there is just one $\gamma$-item, then the total
weight is at most $\frac{k-1}{k}+1 <2$. Otherwise, if $k \leq 19$,
then the total weight is at most $\frac{k-2}{k}+1+0.7-\frac
1k=2.7-\frac 3k$, and if $k \geq 20$, then the total weight is at
most $\frac{k-2}{k}+1+0.65=2.65-\frac 2k \leq 2.7 - \frac 3k$.

Next, consider $\phi$-bins. For $k\geq20$, the proof follows
from the standard analysis \cite{JDUGG74}, and we include it for
completeness. There are at most $k-3$ $\alpha$-items, and their
total weight never exceeds $\frac{k-3}{k}$. If a bin does not
contain a huge $\phi$-item, then it has at most five $\phi$-items
of positive bonuses (each bonus is at most $0.1$), and their
scaled size is at most $1.2$. This gives a total weight of at
most $1-\frac 3k+1.2+0.5=2.7-\frac 3k$. Note that this total
weight cannot be achieved as both situations where there are $k-3$
$\alpha$-items and five $\phi$-items cannot occur simultaneously.
If a bin contains a huge item, then there are at most two
(other) $\phi$-items with positive bonuses. The scaled size of all
$\phi$-items except for the huge item is at most $0.6$, and the total weight excluding the
bonuses of $\phi$-items is at most $2.6-\frac 3k$. If there is only one
$\phi$-item with a positive bonus, then the total weight is at
most $2.7-\frac 3k$ again. Assume that there are two items with
positive bonuses. None of these items can be larger than $\frac
13$, as their total size is below $\frac 12$. If both items have
bonuses of $0.6$ times their sizes minus $0.1$, then their total
bonus is at most $0.6\cdot \frac 12-0.2=0.1$. In the case $k\geq 20$, this is the only remaining option (as each of these items is small or medium), and we are left with the
case $k \leq 19$, and moreover, in the remaining case there are two items with positive bonuses, and these bonuses are not both equal to the sizes times $0.6$ minus $0.1$. Let $a_1 \leq a_2$ be the sizes of
the items. We have $a_2 \in (0.2, 0.3]$ (otherwise either both items are very small, or the larger item is larger medium and the smaller one is very small, and both items have bonuses of the form $0.6$ times the size minus $0.1$, a case that was analyzed earlier).
Thus, the larger item of the two is either larger small or smaller medium.
We will bound the total weight of the two items and
show that it does not exceed $0.7$. Since the weight function is
monotonically non-decreasing, we analyze $w(a_2)+w(\frac 12-a_2)$.
If $\frac 15 < a_2 \leq \frac
14$, then $w(a_2)+w(\frac 12-a_2)=(2.8-20/k)\frac
12-0.7+\frac{10}k=1.4-10/k-0.7+10/k=0.7$. The case $\frac 14 < a_2
\leq 0.3$ is symmetric.
\end{proof}

Now, we bound the total weight of the bins of FF. Once again
we split the analysis into several cases according to the number
of items packed into the bins. In this case we can also neglect
$k$-bins and $1$-bins, as the total weight of a $k$ bin is $1$,
and all items of size above $\frac 12$ are either huge
$\phi$-items, or $\gamma_1$-items. Moreover, any bin that contains
a huge $\phi$-item or a $\gamma_1$-item can be removed from the
analysis. Thus, we are left with $2^+$-bins that do not contain
such items. Additionally, the weight of any bin with level at
least $5/6$ is at least $1$, as the weight of any $\phi$-item and
of a $\gamma_2$-item is at least $6/5$-times the size of the item.
Since there can be at most one $5^{+}$-bin whose level is below
$5/6$, the weight of any $5^{+}$-bin (except for at most one bin)
is at least $1$. In the following we concentrate on the $2$-bins,
$3$-bins and $4$-bins.

\begin{lemma}
The weight of any $2$-bin containing a $\gamma_{2}$-item is at
least $1$, except for at most one bin. The weight of any $3$-bin or $4$-bin, containing a $\gamma_{2}$-item, is at
least $1$, except for at most one bin.
\end{lemma}
\begin{proof}
Assume that at least two bins have $\gamma_2$-items, and each
one has weight below $1$. Denote them by $B_i$ and $B_j$ such that
$B_j$ appears after $B_i$ in the ordering of FF.
Each of them can have at most one $\gamma_2$-item, as the total weight of
two $\gamma_2$ items is above $1$. None of them has a level of at least $\frac 56$, as in such a case the weight is at least $1$.

Assume that both these bins are $2$-bins. The total weight of a $\gamma_2$-item and a $\phi$-item of size above $\frac 14$ is at least $0.35+0.65=1$ for $k \geq 20$, and at least $0.3+\frac 1k+0.7-\frac 1k=1$ for $k \leq 19$. Thus, each of these $2$-bins has a $\phi$-item of size at most $\frac 14$ (as there is only one $\gamma_2$-item packed into each of the two bins). We find $s(B_i)\leq \frac 34$, as the size of the $\gamma_2$-item is at most $\frac 12$, and therefore $B_j$ cannot have any item of size at most $\frac 14$, a contradiction.

Next, assume that $B_i$ and $B_j$ contain $3$ or $4$ items each and have weights below $1$, such that each of them contains one $\gamma_2$-item, and the other items are $\phi$-items. Similarly to the proof for $2$-bins, none of them has a $\phi$-item of size above $\frac 14$. If all the $\phi$-items of $B_{j}$ have sizes of at least
$1/6$, then their total size is at least $\frac 13$, and their total weight is at least $6/5\cdot
\frac 13=4/10$, and we reach a contradiction, since the $\gamma_2$-item of that bin has weight of at least $0.6$. Otherwise, since $B_j$ has an item of size below $\frac 16$, the level of $B_{i}$ is above
$5/6$, a contradiction.
\end{proof}

We are left with bins containing only $\phi$-items that are not huge.
\begin{lemma}
Consider two $\phi$-items of sizes $a_{1}\leq a_{2}\leq1/2$.
If $1 \geq a_{1}+a_{2}>1-a_{1}$ holds, then the total weight of the two items is at
least $1$.
\end{lemma}
\begin{proof}
We have $a_1 > (1-a_2)/2 \geq \frac 14$.
If both items have sizes at least $1/3$, since
$w(1/3)=1/2$, the claim holds, since $w$ is monotonically non-decreasing.
Thus it suffices to consider the case
$1/4<a_{1}\leq1/3$. In this case $a_2 > 1-2a_1 \geq \frac 13$. If $k \geq 20$, then the total weight of the two items is $1.2(a_1+a_2)+(0.6a_1-0.1)+0.1=1.8a_1+1.2a_2=0.9(2a_1+a_2)+0.3a_2 > 0.9\cdot 1+ 0.3 \cdot \frac 13=1$.
If $k \leq 19$, we consider two cases. If $a_1>0.3$, then the calculation is the same as for $ k \geq 20$. Otherwise,
the total weight of the two items is $1.2(a_1+a_2)+((1.6-20/k)a_1-0.4+ \frac 6k)+0.1=(2.8-20/k)a_1+1.2a_2-0.3+ \frac 6k>(2.8-20/k)a_1+1.2(1-2a_1)-0.3+\frac 6k=(0.4-20/k)a_1+0.9+6/k \geq (0.4-20/k)\cdot 0.3+0.9+\frac 6k=1.02>1$, since $0.4-20/k<0$ and $a_1 \leq 0.3$.
\end{proof}

\begin{lemma}
Consider three $\phi$-items of sizes  $a_{1}\leq
a_{2}\leq a_{3}\leq1/2$. If $1 \geq a_{1}+a_{2}+a_{3}>1-a_{1}$ holds, then the total
weight of the three items is at least $1$.
\end{lemma}
\begin{proof}
We have $4a_3 \geq 2a_1+a_2+a_3>1$, so $a_3>\frac 14$.
If $a_1>\frac 14$, then the claim holds since the weight of an item with size above $1/4$ is at least $0.35$.
If $a_1 \leq \frac 16$, then $a_{1}+a_{2}+a_{3}>\frac 56$, and the total weight is at least $1$. Thus, we are left with the case $\frac 16 < a_1 \leq \frac 14$, and thus $a_1+a_2+a_3>\frac 34$. If $a_3>\frac 13$, then its bonus is $0.1$, and the total weight of the three items is at least $\frac 65 \cdot \frac 34+0.1=1$. We are left with the case $\frac 16<a_1 \leq a_2 \leq a_3 \leq \frac 13$.
We find that in the case $k \geq 20$, as all three items have sizes in $(\frac 16, \frac 13]$, the total weight of the items is $1.8(a_1+a_2+a_3)-0.3> 1.8 \cdot \frac 34 -0.3=1.05>1$. We are left with the case $k \leq 19$. If $a_2 > \frac 14$, then since the bonus of any item of size above $\frac 14$ is above $\frac 1{20}$, the total weight of the items is at least $1.2 \cdot \frac 34+2\cdot 0.05=1$. If $a_1 \leq \frac 15$, then $a_{1}+a_{2}+a_{3}>\frac 45$, and since the bonus of the largest item is above $\frac 1{20}$, we get a total weight of at least $1.2 \cdot \frac 45 +0.05=1.01>1$.
We are therefore left with the case $\frac 15 \leq a_1 \leq a_2 \leq \frac 14$.
If $a_3 \leq 0.3$, then the total weight is at least $(2.8-20/k)\frac 34+2(-0.3+4/k)+(-0.4+6/k)=1.1-1/k \geq 1$.
If $a_3 > 0.3$, then the total weight is $(2.8-20/k)(a_1+a_2)+1.8a_3+2(-0.3+4/k)-0.1 > (2.8-20/k)a_1+1.8a_3-0.7+8/k+(2.8-20/k)(1-2a_1-a_3)=(20/k-2.8)a_1+(20/k-1)a_3+2.1-12/k \geq (20/k-2.8)/4+(20/k-1)\cdot 0.3+2.1-12/k=1.1-1/k \geq 1$ since $k \geq 10$, $a_1 \leq \frac 14$, and $a_3 \geq 0.3$.
\end{proof}

\begin{lemma}
Consider four $\phi$-items of sizes $a_{1}\leq a_{2}\leq a_{3}\leq a_{4}\leq1/2$. If $1 \geq a_{1}+a_{2}+a_{3}
+a_{4}>1-a_{1}$ holds, then the total weight of the four items is at least $1$.
\end{lemma}
\begin{proof}
We have $5a_4 \geq 2a_1+a_2+a_3+a_4 >1$, so $a_4 > \frac 15$. If $a_1 \leq \frac 16$, then $a_{1}+a_{2}+a_{3}+a_4>\frac 56$, and the total weight is at least $1$.
Otherwise we find $a_{1}+a_{2}+a_{3}+a_{4} \geq \max\{1-a_1,4a_1\}\geq \frac 45$. If $a_4 > \frac 14$, then its bonus is above $\frac 1{20}$, and the total weight is at least $1.2 \cdot \frac 45+0.05>1$. Thus, $\frac 15 < a_4 \leq \frac 14$. If $k \geq 20$, as the sizes of all items are in $(\frac 16,\frac 14]$, the total weight of all four items is $1.8(a_1+a_2+a_3+a_4)-0.4 \geq 1.04>1$. We are left with the case $k \leq 19$.
If $a_1>\frac 15$, then the total weight of all four items is $(2.8-20/k)(a_1+a_2+a_3+a_4)-1.2+16/k \geq (2.8-20/k)\cdot 0.8-1.2+16/k=1.04 >1$. Otherwise, $\frac 16< a_1 \leq \frac 15 < a_4 \leq \frac 14$. If $a_2 > \frac 15$, then the total weight is $1.8a_1-0.1+(2.8-20/k)(a_2+a_3+a_4)-0.9+12/k >1.8a_1+(2.8-20/k)(1-2a_1)-1+12/k=a_1(40/k-3.8)+1.8-8/k$. If $40/k-3.8$ is non-negative, then using $8/k \leq 0.8$ we find a total weight of at least $1$. Otherwise, using $a_1 \leq \frac 15$, we find a total weight of at least $(40/k-3.8)\cdot \frac 15 +1.8-8/k = 1.04>1$.
If $a_2 \leq \frac 15$, then the scaled size is $1.2(a_1+a_2+a_3+a_4)>1.2(1-a_1)$. Thus bonus of the two smallest items is $0.6(a_1+a_2)-0.2\geq 1.2a_1-0.2$. Thus, the total weight is at least $1$.
\end{proof}
\begin{lemma}
The total weight of the $2$-bins, $3$-bins and $4$-bins of FF that contain only $\phi$-items
is at least their number minus $1$.
\end{lemma}
\begin{proof}
The proof is the same as for Lemma \ref{9last} (the only difference is that $5$-bins are not considered).
\end{proof}

We proved $FF(L)-5 \leq W \leq (2.7-3/k)OPT(L)$.
\begin{theorem}
The asymptotic approximation ratio of $FF$\ for any $k \geq 10$
is at most $2.7-3/k$.
\end{theorem}

\section{Other algorithms}\label{oth}
\subsection{A $2$-competitive algorithm for all $\boldsymbol{k\geq 2}$}
\label{Kot}
Kotov et al. \cite{BCKK04} designed an algorithm that is
$2$-competitive in the asymptotic sense. We present a simplified
version of that algorithm and prove that it is $2$-competitive in
the absolute sense. Our algorithm {\sc Thin and Fat} (TF) has
three kinds of bins.

\noindent 1. Paired bins. Those are bins partitioned into pairs
such that the total size of items for each pair is strictly above
$1$, and the total number of items packed into the two bins is at
least $k$. Those bins will not be used for packing further items.

\noindent 2. Fat bins. Those are bins containing exactly  $k-1$
items.

\noindent 3. Thin bins. Those are non-empty bins containing at
most $k-2$ items.

After we define TF, we will prove that if it has at least one fat
bin, then it has at most one thin bin. The algorithm acts as
follows. Initially all three sets of bins are empty. Let $i \geq
1$ be a new item. The following steps are processed for $i$ until
it is packed.

\noindent 1. If there is a fat bin $B$ such that $s(B)+s_i>1$,
pack $i$ into a new bin, match $B$ and the new bin, these bins
become paired.

\noindent 2. If there are no thin bins, pack $i$ into a new bin.

\noindent 3. If there exists a thin bin $B$ such that $s(B)+s_i
\leq 1$, then pack $i$ into $B$. If $B$ becomes fat and there is a
thin bin $B'\neq B$, match $B$ and $B'$, these bins become paired.

\noindent 4. If there are no fat bins, pack $i$ into a new bin.

\noindent 5. Pack $i$ into a fat bin $B$, match $B$ with a thin
bin $B'$, these bins become paired.

\begin{lemma}
i. In all cases, the actions described above can be performed.

ii. Every two thin bins have total sizes above $1$.

iii. Every two bins that are
matched have sufficient total sizes and sufficient numbers of
items.

iv. If there is at least one fat bin, then there is at most one
thin bin.
\end{lemma}
\begin{proof}
We start with proving part i, i.e., we show that any item $i$ can be packed into the bin
that it is assigned to. In steps 1,2, and 4, the item is packed
into a new bin. Step 3 is applied provided that $B$ exists. Such a
bin has at most $k-2$ items, and has sufficient space. Assume that
step 5 is reached. Since $i$ is not packed in step 1, every fat
bin can receive $i$ since it has $k-1$ items and sufficient space.
Since step 4 was not applied, a fat bin must exist. The only other
action that is performed unconditionally is matching bins in step
5. The thin bin must exist as $i$ was not packed in step 2.

Next we consider part ii. Note that a thin bin can become fat, but a fat bin cannot become
thin. Thus, a new thin bin is created only by packing an item into
a new bin. The bin remains thin as long as it is not paired, and
it has at most $k-2$ items. Obviously, the total size packed into
a bin cannot decrease over time. Thus, to prove this part, it is sufficient to prove
that a pair of a thin bin and a thin bin that was just created and
was not paired immediately have a total size of items above $1$. A
new bin $B'$ that it not paired immediately can be created in
steps 2 and 4. In step 2 it becomes the only thin bin. In step 4,
it is created since the item was not packed in step 3, thus for
any existing thin bin $B$, $s(B)+s(B')=s(B)+s_i>1$.

Consider part (iii). Bins  are matched only  in steps 1,3,5. In step $1$, the two bins
will have $k$ items and total size above $1$. In step 3, the pair
is created only if a thin bin $B$ becomes fat, i.e., $B$ now has
$k-1$ items and $B'$ has at least one item. Moreover, since $B$
and $B'$ were thin, their total size of items is above $1$. In
step 5, since step 3 was not applied, we have $s(B')+s_i>1$, and
the total number of items is at least $k+1$.

To verify the last property (part iv), consider the cases where $i$ is
packed into a bin that is not paired immediately. In step 2, there
will be a unique thin bin. In step 3, if $B$ becomes fat and it is
not paired, then no thin bins remain. Otherwise, there is no
change in the numbers of fat and thin bins. In step 4, there will
be no fat bins after $i$ is packed.
\end{proof}

\begin{theorem}
For any input $L$, $TF(L) \leq 2\cdot OPT(L)$.
\end{theorem}
\begin{proof}
Let $p$, $f$, and $t$ be the numbers of paired bins, fat bins, and
thin bins when the algorithm terminates. Let $W$ be the total size
of items, and $n$ their number.

Assume first that $f=0$. If $t \leq 1$, and $p=0$, then
$TF(L)=t=OPT(L)$. Otherwise, if $p>0$ (and $t\leq 1$), we find that $W >  \frac{p}2$. Therefore
$OPT(L) \geq W > \frac p2$ and $OPT(L) \geq \frac{p+1}2$. Thus,
$p+f+t \leq p+1 \leq 2 OPT(L)$. If $t \geq 2$, since the total
size of items of every pair of thin bins is above $1$, we have
$OPT(L) \geq W > \frac p2 + \frac t2$, while $TF(L)=p+t$.

Otherwise, $f \geq 1$. In this case $t \leq 1$. We have $n \geq \frac p2 \cdot k + (k-1)f+t$. If $t=0$, then we get $n \geq \frac k2 (p+f)=\frac k2 TF(L)$, since $k-1 \geq \frac k2$. Otherwise, $n \geq \frac p2 \cdot k + (k-1)f+1=\frac k2(p+f+1)+(\frac k2-1)f-\frac k2+1=\frac k2(p+f+1)+(\frac k2-1)(f-1) \geq \frac k2(p+f+1) =\frac k2 TF(L)$. Since $OPT(L) \geq \frac nk$, and we get $TF(L) \leq 2OPT(L)$.
\end{proof}

\subsection{A simple algorithm with an absolute competitive ratio $\boldsymbol{2}$ for $\boldsymbol{k =5}$}
We present a different algorithm that is based on an adaptation of
FF. The algorithm $ALG$ acts as follows. A new item $i$ is
assigned into a minimum index bin $B$ that satisfies all the
following conditions (if no such bin exists, then it is assigned
into an empty bin).

\noindent 1. The current level of $B$ is at most $1-s_i$.
\noindent 2. The current number of items of $B$ is at most $4$.
\noindent 3. If $B$ currently contains four items, then after
assigning $i$, its level will be at least $\frac 12$.

A regular bin is a bin of $ALG$ that is a $2$-bin or a $3$-bin. We
treat $1$-bins (also called dedicated bins), $4$-bins, and
$5$-bins separately.

\begin{lemma}\label{basics}
The level of any $5$-bin is at least $\frac 12$. The level of any
regular bin, except for at most one bin is at least $\frac 23$.
For a pair of bins $B$, $B'$ (where $B'$ appears after $B$ in the
ordering), the total size of items packed into $B$ and $B'$ is
above $1$ in the following two cases.

1. None of the bins contains more than three items. 2. One of the
bins is a large dedicated bin, and the other contains four items.
\end{lemma}
\begin{proof}
By the third condition, $ALG$ never creates a $5$-bin whose level
is below $\frac 12$. For regular bins,  $ALG$ simply applies FF
(without cardinality constraints) on the subsequence of items of
these bins, and thus the claim follows from Claim
\ref{veryuseful}.

If $B$ contains at most three items, then when the first item of
$B'$ is packed, the third condition is irrelevant, the second
condition holds, and thus the first condition does not hold, and
the total size of items packed into the two bins exceeds $1$.

We are left with the case that $B$ contains four items when the
unique item $i$ of $B'$ arrives, and $s_i  > \frac 12$ (the case that $B'$ has four items but $B$ has one item was already considered). Thus, if
$i$ is packed into $B$, the third condition must hold, and
therefore the first condition does not hold, and the total size of
items packed into the two bins exceeds $1$.
\end{proof}

Consider an input $L$. For the output of ALG applied on $L$, let
$f$ denote the number of $4$-bins, let $d_1$ be the number of
dedicated bins whose items have sizes above $\frac 12$ (such bins
and items are called large dedicated bins and large dedicated
items), and let $d_0$ be the number of dedicated bins whose items
have sizes no larger than $\frac 12$ (such bins and items are
called small dedicated bins and small dedicated items). By Lemma
\ref{basics}, $d_{0}\leq1$ must hold.

\begin{theorem}
\label{k5} The absolute competitive ratio of $ALG$ is at most $2$.
\end{theorem}
\begin{proof}
We distinguish two cases as follows.

\noindent \textbf{Case 1. $\boldsymbol{f<d_{1}}$}. We match
$4$-bins and large dedicated bins into pairs arbitrarily, leaving
at least one large dedicated bin unmatched. The remaining bins
that are not $k$-bins (regular bins, and a small dedicated bin, if
it exists) are also matched into pairs, and if the number of these
bins is odd, one of them is matched to an unmatched  large
dedicated bin. The total size of items of any matched pair is
above than $1$, the level of every remaining large dedicated bin
is above $\frac 12$, and the level of any $k$-bin is at least
$\frac 12$, by Lemma \ref{basics}. We find that the total size of
items $S$ satisfies $S \geq \frac {ALG}2$, thus $OPT(L) \geq S
\geq \frac{ALG}2$.

\noindent \textbf{Case 2. $\boldsymbol{f \geq d_{1}}$}. For an
item of size $x$, we define its weight to be $w(x)=1+3x$. Let $W$
denote the total weight of all items of $L$. For any bin, the
total weight of its items is at most $8$, as it has at most five
items of a total size of at most $1$. Match every large dedicated
bin to a $4$-bin. For each such pair, the total size is above $1$,
and the number of items is $5$, thus the total weight of the items
of every such pair of bins is at least $8$. Every remaining
$4$-bin has four items, and their total weight is at least $4$.
Similarly, every $5$-bin has a total weight above $5$. Every
regular bin, except for at most one such bin, has a total size of
items of at least $\frac 23$, and at least two items, so its
weight is at least $2+\frac 23 \cdot 3=4$. Thus, on average, all
bins have weights of at least $4$, except for possibly a small
dedicated bin, if it exists, and a regular bin of level below
$\frac 23$, if it exists. If none of those exists, we find $8\cdot
OPT(L) \geq W \geq 4 \cdot ALG(L)$, and we are done. If both such
bins exist, then the total size of items in those bins is above
$1$ by Lemma \ref{basics}, and there are at least three items, so
the total weight is at least $6$. We find $ W \geq 4 \cdot
(ALG(L)-2)+6$. If $d_0=1$ but no regular bin of level below $\frac
23$ exists, then the weight of the small dedicated bin is at least
$1$, and $W \geq 4 \cdot (ALG(L)-1)+1$. If $d_0=0$, but there is a
regular bin of level below $\frac 23$, then the weight of this
regular bin is at least $2$, and $W \geq 4 \cdot (ALG(L)-1)+2$. In
all three last cases, $8 OPT(L) \geq W \geq 4\cdot ALG(L)-3$, or
alternatively, $ALG(L) \leq 2 \cdot OPT(L)+3/4$. By integrality,
we get $ALG(L) \leq 2 \cdot OPT(L)$.
\end{proof}

\bibliographystyle{abbrv}

\end{document}